\documentclass[journal]{IEEEtran}
\usepackage{amssymb,bm,graphicx,float,caption,bbm,amsmath,breqn,multirow,amsthm,url,cite,comment,cleveref,xcolor,epstopdf}
\newtheorem{theorem}{Theorem}
\newtheorem{lemma}{Lemma}
\theoremstyle{definition}
\newtheorem{defn}{Definition}
\newtheorem{exmp}{Example}

\newtheorem{cor}{Corollary}

\begin{document}

\title{Incentive Compatibility in Stochastic Dynamic Systems}

\author{Ke~Ma,~\IEEEmembership{Member,~IEEE,}
	P.~R.~Kumar,~\IEEEmembership{Fellow,~IEEE}
\thanks{This work was supported in part by NSF Contract ECCS-1760554, NSF Science \& Technology Center Grant CCF-0939370, the Power Systems Engineering Research Center (PSERC), and NSF Contract IIS-1636772.}
\thanks{Ke Ma is with Pacific Northwest National Laboratory, Richland, WA 99352 and P. R. Kumar is with Texas A\&M University, College Station, TX 77843-3259 (e-mail: ke.ma@pnnl.gov, prk.tamu@gmail.com)}
	}

\maketitle
\begin{abstract}
The classic Vickrey-Clarke-Groves (VCG) mechanism ensures incentive compatibility, i.e., truth-telling is a dominant strategy for all agents, for a static one-shot game. However, it does not appear to be feasible to construct
mechanisms that ensure dominance of dynamic truth-telling for agents comprised of general stochastic dynamic systems. The agents' intertemporal net utilities depend on future controls and payments, and a direct extension of the VCG mechanism does not guarantee incentive compatibility. This paper shows that such a stochastic dynamic extension does exist for the special case of Linear-Quadratic-Gaussian (LQG) agents. In fact it achieves subgame perfect dominance of truth telling. This is accomplished through a construction of a sequence of layered payments over time that decouples the intertemporal effect of current bids on future net utilities if system parameters are known and agents are rational.

An important motivating example arises in power systems where an Independent System Operator has to ensure balance of generation and consumption at all times, while ensuring social efficiency, i.e., maximization of the sum of the utilities of all agents. It is also necessary to satisfy budget balance and individual rationality. However, in general, even for static one-shot games, there is no mechanism that simultaneous satisfies these requirements while being incentive compatible and socially efficient. For a power market of LQG agents, we show that there is a modified ``Scaled'' VCG (SVCG) mechanism that does satisfy incentive compatibility, social efficiency, budget balance and individual rationality under a certain ``market power balance" condition where no agent is too negligible or too dominant. 

We further show that the SVCG payments converge to the Lagrange payments, defined as the payments that correspond to the true price in the absence of strategic considerations, as the number of agents in the market increases.

For LQ but non-Gaussian agents, optimal social welfare over the class of linear control laws is achieved.
\end{abstract}

\section{introduction}
Mechanism design is the sub-field of game theory that considers how to realize socially optimal solutions to problems involving multiple self-interested agents, each with a private valuation for each outcome. This valuation is generally represented as utility function, which captures the monetary value of each alternative. A typical approach is to provide financial incentives such as payments to promote truth-telling of utility function parameters by agents. An important example is the Independent System Operator (ISO) problem of electric power systems in which the ISO aims to maximize social welfare and maintain balance of generation and consumption while each generator/load has a private utility function. 

The classic Vickrey-Clarke-Groves (VCG) mechanism \cite{JOFI:JOFI2789}, and its generalization, the Groves mechanism \cite{10.2307/1914085}, play a central role in mechanism design since they ensure incentive compatibility, i.e., they ensure that truth-telling of utility functions by all agents forms a dominant strategy, as well as social welfare optimality, i.e., the sum of utilities of all agents is maximized. The outcome generated by the Groves mechanism is stronger than a Nash equilibrium in the sense that it is \emph{strategy-proof}, meaning truth-telling of utility functions is optimal irrespective of what others are bidding. In fact, Green and Laffont \cite{Green1979} show that the Groves mechanism is the only mechanism that is both efficient and strategy-proof if net utilities are quasi-linear, i.e., linear in the amount of money.

While the Groves mechanism is applicable to a static one-shot game, it does not work for stochastic dynamic games where agents are given the opportunity to bid utility functions and states at each time instant. In a stochastic dynamic system that unfolds over time, the agents' intertemporal net utilities depend on the future controls and payments, and a direct extension of the VCG mechanism does not guarantee incentive compatibility \cite{RePEc:cwl:cwldpp:1757r}. A fundamental difference between dynamic and static mechanism design is that in the former, an agent can bid an untruthful utility function conditional on its past bids (which need not be truthful) and past allocations (from which it can make an inference about other agents' utility functions) \cite{RePEc:cwl:cwldpp:2102}. For dynamic \emph{deterministic} systems, by collecting the VCG payments as a lump sum of all the payments over the entire time horizon at the beginning, incentive compatibility is still assured. However, for dynamic \emph{stochastic} systems, the states are private random variables and it is necessary to incentivize agents to bid their states truthfully. It does not appear to be feasible to construct mechanisms that ensure the dominance of dynamic truth-telling for agents comprised of general stochastic dynamic systems. Indeed we conjecture that it is not possible to do so for general stochastic dynamic agents.

Our contribution, on which we build further desirable results as described in the sequel, is to show that for the special case of Linear-Quadratic-Gaussian (LQG) agents, where agents have linear state equations, quadratic utility functions and additive white Gaussian noise, a dynamic stochastic extension of the VCG mechanism does exist, based on a careful construction of a sequence of layered payments over time. We propose a modified layered mechanism for payments that decouples the intertemporal effect of current bids on future net utilities, and prove that truth-telling of their dynamic states is a dominant strategy for every agent, if system parameters are known and agents are rational. ``Rationality" is defined in a dynamic programming fashion: an agent is rational at the last time instant if it adopts a dominant strategy whenever there exists a unique one; an agent is rational at time $t$ if it adopts a unique dominant strategy, assuming that all agents are rational at all future times. In fact the mechanism achieves subgame perfect dominance of truth-telling.

An important example of a problem needing such optimal dynamic coordination of stochastic agents arises in the ISO problem of power systems. Renewable energy resources such as solar/wind are stochastic and dynamic in nature, as are consumptions by loads which are influenced by factors such as local temperatures and thermal inertias of facilities. References \cite{6713053} and \cite{osti_15011696} provide details on how wind turbines can be modeled as LQG systems. In general, agents may have different approaches to responding to the prices set by the ISO. If each agent acts as a \emph{price taker}, i.e., it honestly discloses its energy consumption/production at the announced prices, a \emph{competitive equilibrium} would be reached among agents. However, if agents are \emph{price anticipators}, then it is critical for the ISO to design a market mechanism that is strategy-proof (i.e., incentive compatible). The challenge for the ISO is to determine a bidding scheme between agents (producers and consumers) and the ISO that maximizes social welfare, while taking into account the stochastic dynamic models of agents. 

Currently, the ISO solicits bids from generators and Load Serving Entities (LSEs) and operates two markets: a day-ahead market and a real-time market. The day-ahead market lets market participants commit to buy or sell wholesale
electricity one day before the operating day, to satisfy energy demand bids and
to ensure adequate scheduling of resources to meet the next
day's anticipated load. The real-time market lets market
participants buy and sell wholesale electricity during the
course of the operating day to balance the differences between the already pledged 
day-ahead commitments and the actual real-time demand and
production \cite{isone}. Our layered VCG mechanism fits perfectly in the real-time market since at each time instant the market clears and settles, as each agent bids its random state after realization.

However, there is also a potential fatal downside for the VCG mechanism: in general, the sum of total payments collected by the ISO is non-positive, i.e., there is a \emph{budget deficit}. In fact, when agents have quadratic utility functions, the total payments collected from consumers is indeed not enough to cover the total payments to the suppliers. In effect, in order to force agents to reveal their true utility functions, the ISO will need to subsidize the market. In this paper we will also propose a solution to this problem. The VCG payment charges each agent $i$ the difference between social welfare of others if agent $i$ is absent and social welfare of others when agent $i$ is present. We exhibit a solution that ensures there is no budget deficit. It consists of inflating the first term above in all of the agents' VCG payments by a constant factor $c$, leading to a Scaled VCG (SVCG) mechanism. 

There are however two additional issues to be addressed when proposing such a scheme. The first concerns the issue of \emph{individual rationality}. The \emph{net utility} that an agent obtains is the utility of energy consumption minus the amount it pays. The magnitude of the scaling factor $c$ is important because an agent may opt out of the process if $c$ is chosen to be too large. The reason lies in the fact that not joining the market results in a net utility of zero, obtained from generating/consuming no power and collecting/making no payments, while joining the market results in a negative net utility. That is, the scheme is not individually rational. The second issue concerns whether the payment
is \emph{Lagrange optimal} for each agent. In power markets, electricity prices are calculated as the sum of all Lagrange multipliers associated with different constraints, and payments are collected as the product of price and power quantity. By Lagrange optimality is meant the price that would manifest and the
payments that would occur if all agents were truth tellers. The concern is that if a customer participates, the price it pays/receives need not be Lagrange optimal. 

We show that under a ``market power balance" condition, which essentially requires that no agent is too negligible or too powerful, there is indeed a systematic way to choose the scaling factor number $c$ such that there is no budget deficit for the ISO, while at the same time guaranteeing that producers and consumers will actively participate in the market. The factor $c$ can be chosen in a way such that the maximum distortion between the VCG payment and Lagrange payment is minimized. We argue that based on historic knowledge of the market, the ISO may be able to choose such a $c$ that does not depend on any agent's tactical announcement. 

Moreover, we show that asymptotically, as the number of agents increases, the Scaled VCG payments converge to the Lagrange optimal payments. This result provides an economic justification for Load Serving Entities, which aggregate a large number of small consumers into one reasonable sized consumer, as being required to achieve social welfare optimality.

In Section \ref{RWCDC}, a survey of related works is presented. Section \ref{SDCDC} provides a description of the classic VCG framework for the static and dynamic deterministic problems. It also introduces the Scaled VCG mechanism for individual rationality and budget balance. Section \ref{SCDC} presents the difficulties in designing a mechanism that ensures dominance of truth telling when agents are general stochastic dynamic systems. Section \ref{LQGS} presents the layered mechanism for LQG systems. It shows how intertemporal decoupling is achieved, and proves subgame perfect dominance of truth telling. It also shows that the budget balance and individual rationality properties of the scaling mechanism carry over from the deterministic case. A numerical example is presented. Section \ref{NG} shows how the mechanism also applies when the noises are not Gaussian, to achieve optimality in the class of linear feedback strategies. Section \ref{CRCDC} concludes the paper.

\section{Related Works}\label{RWCDC}

In recent years, several works have explored issues arising in dynamic mechanism design. In order to achieve ex-post incentive compatibility, Bergemann and Valimaki \cite{RePEc:cwl:cwldpp:1616} propose a generalization of the VCG mechanism based on the marginal contribution of each agent and show that ex-post participation constraints are satisfied under some conditions. Athey and Segal \cite{ECTA:ECTA1375} consider an extension of the d'Aspremont-Gerard-Varet (AGV) mechanism \cite{DASPREMONT197925} to design a budget balanced dynamic incentive compatible mechanism. Pavan et al. \cite{Pavan2009Dynamic} derives first-order conditions under which incentive compatibility is guaranteed by generalizing Mirrlees's \cite{10.2307/2296779} envelope formula of static mechanisms. Cavallo et al. \cite{DBLP:journals/corr/CavalloPS12} considers a dynamic Markovian model and derives a sequence of Groves-like payments which achieves Markov perfect equilibrium. Bapna and Weber \cite{Bapna2005EfficientDA} solves a sequential allocation problem by formulating it as a multi-armed bandit problem. Parkes and Singh \cite{NIPS2003_2432} and Friedman and Parkes \cite{Friedman:2003:PWS:779928.779978} consider an environment with randomly arriving and departing agents and propose a ``delayed'' VCG mechanism to guarantee interim incentive compatibility. Besanko et al. \cite{BESANKO198533} and Battaglini et al. \cite{RePEc:pri:metric:wp046_2012_battaglini_lamba_optm_dyn_contract_10october2012_short.pdf} characterize the optimal infinite-horizon mechanism for an agent modeled as a Markov process, with Besanko considering a linear AR(1) process over a continuum of states, and Battaglini focusing on a two-state Markov chain. Bergemann and Pavan \cite{BERGEMANN2015679} have an excellent survey on recent research in dynamic mechanism design. A more recent survey paper by Bergemann and Valimaki \cite{RePEc:cwl:cwldpp:2102} further discusses the dynamic mechanism design with risk-averse agents and the relationship between dynamic mechanism and optimal contracts.  

In order to capture strategic interactions between the ISO and market participants, game theory and mechanism design has been proposed in many recent papers. Sessa et al. \cite{DBLP:journals/corr/SessaWK16} studies the VCG mechanism for electricity markets and derives conditions to ensure collusion and shill bidding are not profitable. Okajima et al. \cite{7320631} propose a VCG-based mechanism that guarantees incentive compatibility and individual rationality for day-ahead market with equality and inequality constraints. Xu et al. \cite{7297853} shows that the VCG mechanism always results in higher per-unit electricity prices than the locational marginal price (LMP) mechanism under any given set of reported supply curves, and that the difference between the per-unit prices resulting from the two mechanisms is negligibly small. Bistarelli et al. \cite{DBLP:journals/corr/BistarelliCGM16} derives a VCG-based mechanism to drive users in shifting energy consumption during peak hours. In Samadi et al. \cite{6102349}, it is proposed that utility companies use VCG mechanism to collect private information of electricity users to optimize the energy consumption schedule. Taylor et al. \cite{6669505} formulates the regulation pricing policy as an LQR problem and applies the VCG mechanism to induce honest participation.

There are also some related works aiming at achieving budget balance for VCG mechanism. Parkes et al. \cite{Parkes:2001:ABV:1642194.1642250} designs a budget-balanced and individual-rational mechanism for combinatorial exchanges that sacrifices incentive compatibility. Moulin et al. \cite{Moulin2001} discusses the trade-off between budget balance and efficiency of the mechanism. Cavallo \cite{Cavallo2006} uses domain information regarding agent valuation spaces to achieve redistribution of much of the required transfer payments back among the agents. Similarly, Thirumulanathan et al. \cite{Thirumulanathan2017} propose a mechanism that is efficient and comes close to budget balance by returning much of
the payments back to the agents in the form of rebates. In \cite{Ma2014}, an enhanced (Arrow-d’Aspremont-Gerard-Varet) AGV mechanism is proposed to tackle the problem of budget balance in demand side management. Karaca et al. \cite{DBLP:journals/corr/abs-1811-09646} proposes core-selecting mechanisms that are coalition-proof and budget balanced, but only with approximate incentive compatibility, i.e., the sum of potential profits of each bidder from a unilateral deviation is minimized. Tanaka et al. \cite{8430852} shows that payments calculated by the clearing-price mechanism and VCG mechanism are similar when each individual agent's market power is negligible. In \cite{NBERw23771}, the authors propose a criterion of approximate strategy-proofness called strategy-proofness in the large (SP-L) for a large market identifying a Lagrange multiplier based mechanism as SP-L. Our work, on the other hand, shows that payments defined in the strategy-proof VCG mechanism converges to Lagrange payments if the market is sufficiently large.

The problem of how to conduct bidding to achieve social welfare optimality in stochastic dynamic systems is examined in \cite{8274944}. It assumes that all agents are price-takers.

A preliminary announcement of some of these results was presented in the conference paper \cite{DBLP:journals/corr/abs-1803-06734}. The layered payment structure for LQG systems is mentioned there, and incentive compatibility results are presented without proofs. The present paper contains the complete proof of incentive compatibility, and further introduces the Scaled VCG mechanism for budget balance, individual rationality, and Lagrange optimality.

To our knowledge, there does not appear to be any result that ensures dominance of dynamic truth-telling for agents comprised of LQG systems, let alone ensuring no budget deficit for the ISO, and individual rationality for all agents.
\section{Static and Dynamic Deterministic Systems}\label{SDCDC}

\subsection{Static Deterministic Systems} 
We begin with the simpler static deterministic case. There are $N$ agents, each having a utility function $F_{i}(u_{i})$, where $u_{i}$ is the amount of energy produced/consumed by agent $i$, with $u_{i}\le 0$ for a producer and $u_{i}\ge 0$ for a consumer. Let $\boldsymbol{u}:=(u_{1},...,u_{N})^{T}$, $\boldsymbol{u}_{-i}:=(u_{1},...,u_{i-1},u_{i+1},...,u_{N})^{T}$. 
The ISO wishes to maximize the \emph{social welfare}, $\sum_{i} F_{i}(u_{i})$: 
\begin{subequations}\label{utilitymax}
\begin{align}
\max_{\boldsymbol{u}}\sum_{i} F_{i}(u_{i}), \tag{\ref{utilitymax}}\end{align} 
\begin{equation}
\text{subject to } \sum_{i} u_{i}=0. \label{umaxcon}
\end{equation}
\end{subequations}
Without loss of generality we only include power balance constraint \eqref{umaxcon}. Other constraints such as DC power flow constraints and capacity constraints can be added and will not change the mechanism structure.

The ISO does not know the individual utility functions of the agents. If it asks them to disclose their utility functions, they may lie in order to obtain a better allocation. A solution to this problem of ``truth-telling" is provided by the VCG mechanism which asks each agent to bid its utility function. Denote agent $i$'s bid by $\hat{F}_{i}$ and let $\hat{F} := (\hat{F}_{1}, \ldots ,\hat{F}_{n})$. After obtaining the bids, the ISO calculates $\boldsymbol{u^{*}}(\hat{F})$ as the optimal solution to:
\begin{equation*}
\max_{\boldsymbol{u}}\sum_{i}\hat{F}_{i}(u_{i}),\text{ subject to }\eqref{umaxcon}.
\end{equation*}
Each agent is then assigned to produce/consume $u_{i}^{*}(\hat{F})$, accruing a utility $F_{i}\left(u_{i}^{*}\left(\hat{F}\right)\right)$. Following the rule that it had announced a priori before receiving the bids, the ISO collects a \emph{payment} $p_{i}(\hat{F})$ from agent $i$, defined as:
\begin{equation*}
p_{i}(\hat{F}):=\sum_{j\ne i}\hat{F}_{j}(\boldsymbol{u}^{(i)})-\sum_{j\ne i}\hat{F}_{j}(\boldsymbol{u^{*}}(\hat{F})),
\end{equation*} 
\noindent where $\boldsymbol{u}^{(i)}$ is defined as the optimal solution to:
\begin{equation*}
\max_{\boldsymbol{u}_{-i}}\sum_{j\ne i}\hat{F}_{j}(u_{j})\text{, subject to }\sum_{j\ne i}u_{j}=0.
\end{equation*}
The VCG mechanism is a special case of the Groves mechanism \cite{10.2307/1914085}, where payment $p_{i}$ is defined as:
\begin{equation*}
p_{i}(\hat{F}):=h_{i}(\hat{F}_{-i})-\sum_{j\ne i}\hat{F}_{j}(\boldsymbol{u}^{*}(\hat{F})),
\end{equation*}
where $h_{i}$ is any arbitrary function of $\hat{F}_{-i}:=(\hat{F}_{1},..,\hat{F}_{i-1},\hat{F}_{i+1},...,\hat{F}_{N})$. 

Define the ``net utility" of an agent as the utility derived by it minus its payment. Truth-telling is a \emph{dominant strategy} in the Groves mechanism \cite{10.2307/1914085}, i.e., each agent maximizes its net utility by bidding its true utility function, regardless of what other agents bid:
\begin{align*}
F_{i}(u^{*}(\bar{F})) - p_i(\bar{F}) \geq &F_{i}(u^{*}(\hat{F})) - p_i(\hat{F}),\text{ for all }\hat{F}_{i},
\end{align*}
where $\bar{F}:=(\hat{F}_{1}, \hat{F}_{2}, \ldots , \hat{F}_{i-1}, F_{i} , \hat{F}_{i+1}, \ldots , \hat{F}_{N})$.

\theoremstyle{definition}
\begin{defn}
	A mechanism is \emph{incentive compatible (IC)} if truth-telling is a dominant strategy for every agent.
\end{defn}
One should note that an agent may not necessarily tell the truth even if truth-telling is dominant since there may be another strategy that is also dominant. We will assume that every agent is ``rational," in that if its dominant strategy is \emph{unique}, then an agent will indeed tell the truth.
\begin{defn}
	We call a mechanism \emph{efficient (EF)} if the resulting allocation $\boldsymbol{u^{*}}$ maximizes the social welfare $\sum_{i}F_{i}(u_{i})$.
\end{defn} 

Two more important properties are sought in a solution.
\begin{defn}
	A mechanism is ex-post \emph{individually rational (IR)} if agents are guaranteed to gain a nonnegative \emph{net utility} by participating, that is, $F_{i}(u_{i}^{*})-p_{i}\geq 0$ (because, by abstaining, an agent can always realize a net utility of zero).
\end{defn}
\begin{defn}
	A mechanism satisfies \emph{budget balance (BB)} if the total payment made by agents is nonnegative: $\sum_{i}p_{i}\ge 0$. (We use the descriptor ``budget balance" if the ISO does not have to provide a subsidy, rather than strictly requiring that the total payments are exactly zero).
\end{defn}
The VCG mechanism, in general, does not satisfy
BB. In fact, Green and Laffont \cite{Green1979} show that no mechanism can satisfy all four properties --IC, EF, IR \& BB-- at the same time. 

If the ISO knew the true utility functions of all the agents, it could solve the social welfare problem in a centralized manner: calculate the Lagrange multiplier $\lambda^{*}$ (price) for the constraint \eqref{umaxcon}, and collect a payment equal to $\lambda^{*}u_{i}^{*}$ from agent $i$. We call this the \emph{Lagrange payment}:
\begin{defn}
	Suppose the optimal solution $(\lambda^{*},u^{*})$ to \eqref{utilitymax} is unique. We say that a mechanism is \emph{Lagrange Optimal} if the payment $p_{i}$ collected from agent $i$ is equal to $\lambda^{*}u_{i}^{*}$.
\end{defn}

We note that the Lagrange payment scheme is widely used in current ISO operation, since the ISO solves the social welfare problem with approximate DC power flow. When non-convexities such as AC power flow are incorporated, the Lagrange payment is seldom used because of the duality gap. It is also worth noting that with only constraint \eqref{umaxcon}, the sum of all Lagrange payments is exactly zero. However, budget balance generally does not hold when other constraints such as DC power flow are included\cite{LiTesfatsion2011, Alderete2005}.

We show in the sequel that while there is no mechanism that satisfies all four properties (IC, EF, IR and BB) \emph{in general}, there does exist such a mechanism under a certain ``Market Power Balance" condition. We inflate (or deflate) the first term in the standard VCG mechanism by a \emph{constant} factor $c$:
\begin{equation}\label{cdefn}
p_{i}(\hat{F})=c\cdot \sum_{j\ne i}\hat{F}_{j}(\boldsymbol{u}^{(i)})-\sum_{j\ne i}\hat{F}_{j}(\boldsymbol{u^{*}}).
\end{equation}
We call this a Scaled VCG (SVCG) mechanism, and $c$ as the scaling factor. To achieve BB and IR, one could choose $c$ as a function of the bids $\hat{F}_{i}$, which unfortunately would cease to guarantee incentive compatibility since the first term in $\eqref{cdefn}$ is not allowed to be dependent on $\hat{F}_{i}$ in the Groves mechanism. 

We show below that there is a range of values of $c$ that can ensure BB under a certain market power balance condition, and argue that through its long-term operation, the ISO may be able to learn at least a subset of this range. Our presumptive argument rests on the repetitive nature of this problem which is played out every day, allowing the ISO to be able to tune $c$ to avoid a net subsidy. Based on this experience, the ISO could choose a $c$ for which BB and IR hold at the same time. However, agents may play a game on the constant $c$ after interacting with the ISO for an extended period of time. Incentives or schemes to prevent such play, which occurs over a slower time-scale, are an important future research direction.

\begin{theorem}\label{c}
Let $\boldsymbol{u}^{*}$ be the optimal solution to \eqref{utilitymax} and suppose that $u^{*}$ is unique. Suppose that $\boldsymbol{u}^{(i)}$ is the unique optimal solution to:
\begin{equation*}
\max\ \sum_{j\ne i}F_{j}(u_{j}), \text{ subject to }\sum_{j\ne i}u_{j}=0.
\end{equation*}

	Let $H_{i}:=\sum_{j\ne i}F_{j}(\boldsymbol{u}^{(i)})$, and let $H_{max}:=\max_{i}H_{i}$. Let $F_{total}=\sum_{j}F_{j}(\boldsymbol{u}^{*})$. If $F_{total}>0$, $H_{i}>0$ for all $i$, and the following \emph{Market Power Balance (MPB)} condition holds:
\begin{equation}\label{Market power balance}
(N-1)H_{max}\le \sum_{i}H_{i},
\end{equation}
then there exists an interval $[\underline{c},\bar{c} ]$ such that for any $c$ in this interval, the SVCG mechanism satisfies IC, EF, BB and IR at the same time.
\end{theorem}
\begin{proof}
	With $c$ chosen as a constant, the SVCG mechanism is within the Groves class and thus satisfies IC and EF.
	To achieve budget balance, we need
	\begin{equation*}
	\sum_{i}p_{i}=c\sum_{i}H_{i}-(N-1)F_{total}\ge 0.
	\end{equation*}
	To achieve individual rationality for agent $i$, we also need
	\begin{equation*}
	F_{i}(\boldsymbol{u}^{*})-p_{i}=F_{i}(\boldsymbol{u}^{*})-c\cdot H_{i}+\sum_{j\ne i}F_{i}(\boldsymbol{u}^{*})\ge 0.
	\end{equation*}
	Combining both the inequalities, we need to be able to choose a $c$ such that
	\begin{equation}\label{cbound}
	\frac{(N-1)F_{total}}{\sum_{i}H_{i}}\le c\le\frac{F_{total}}{H_{max}}.
	\end{equation}
	Let $\underline{c}:=\frac{(N-1)F_{total}}{\sum_{i}H_{i}}$, and $\bar{c}:=\frac{F_{total}}{H_{max}}$. Such a $c$ exists if 
	\begin{equation*}
	(N-1)H_{max}\le \sum_{i}H_{i},\ F_{total}>0,\ H_{i}>0,\ \forall i.\qedhere
	\end{equation*}
\end{proof}

The critical condition \eqref{Market power balance} can be interpreted as requiring that no agent has significantly bigger or smaller market power than others. Individual residential load customers generally have a much smaller scale compared to power plants, and it is thus beneficial to form load aggregators or utility companies at the consumer side, as suggested by the MPB condition. This provides an economic justification for the role of load aggregators or load serving entities to guarantee the achievement of social welfare maximization. The MPB condition is one sufficient condition for budget balance and individual rationality that we have been able to identify. It would be of interest to determine if there are looser conditions than \eqref{Market power balance} to guarantee BB and IR.

In general, an SVCG mechanism is however not Lagrange optimal. Within the feasible range $[\underline{c},\bar{c}]$, one may prefer to choose a $c$ that also achieves near-Lagrange optimality. This could be formulated as the following MinMax problem:
\begin{equation*}
\min_{c}\max_{i}\ |d_{i}(c)|,\text{ subject to \eqref{cbound}},
\end{equation*} 
where $d_{i}(c):=\lambda^{*}u_{i}^{*}-p_{i}=\lambda^{*}u_{i}^{*}-c\cdot H_{i}+\sum_{j\ne i}F_{j}(\boldsymbol{u}^{*})$. 
\begin{comment}
This can be transformed to a linear program:
\begin{equation*}
\min\ Z
\end{equation*}
\noindent subject to
\begin{equation*}
Z\ge d_{i}(c), \text{ for all }i,
\end{equation*}
\begin{equation*}
Z\ge -d_{i}(c), \text{ for all }i,\text{ and }\eqref{cbound}.
\end{equation*}
\end{comment}

\begin{exmp}
	All agents have quadratic utility functions: $F_{i}=r_{i}u_{i}^{2}+s_{i}u_{i}$. $(r_{1},r_{2},r_{3},r_{4})=(-1,-1.1,-1.2,-1.1)$ and $(s_{1},s_{2},s_{3},s_{4})=(1,1.2,4,5)$. The unique Lagrange optimal solution is $\boldsymbol{u}^{*}=(-0.86,-0.70,0.53,1.03)$, $\lambda^{*}=2.73$, and from \eqref{cbound}, $1.13\le c\le 1.19$. The optimal solution to the MinMax problem is $(c^{*},Z^{*})=(1.14,0.22)$. By choosing $c=1.14$, the SVCG mechanism satisfies IC, EF, BB and IR, and the maximum discrepancy between VCG and Lagrange payments is 0.22.
	
	In the MinMax problem, one can also replace $d_{i}(c)$ by the ratio $d_{i}(c)/\lambda^{*}u_{i}^{*}$ to normalize the nearness of the payment to the Lagrange payment. It also can be written as an LP. Using the above ratio, the optimal solution is $(c^{*},Z^{*})=(1.18,0.06)$, showing that all agents pay/receive within $6\%$ of their Lagrange optimal payment.
	
	If we change $s_{2}$ to $1.8$ (a higher marginal cost which implies smaller market power) while keeping the remaining parameters unchanged, the optimal solution becomes $(c^{*}, Z^{*})=(1.13,0.144)$, with distortions $(6.5\%,14.4\%,14.3\%,7.2\%)$. Therefore, unevenly distributed market power increases the maximum distortion between the SVCG payment and Lagrange payment, with a smaller agent facing higher distortion.
	
	We also note that if we further increase $s_{2}$ to $2.0$, the MPB condition is violated and hence we cannot find a constant c that satisfies both budget balance and individual rationality.
\end{exmp}
\subsection{Static Deterministic Systems with Quadratic Costs: Asymptotics}
We consider $N$ heterogeneous agents with quadratic costs and show that SVCG payments converge to the Lagrange payment as $N$ increases. Let $F_{i}(u_{i})=a_{i}u_{i}^2+b_{i}u_{i}$ be the concave quadratic utility function for agent $i$, whether supplier or consumer. Denote by $A=diag(a_{1}, a_{2},\ldots,a_{N})$ the diagonal matrix consisting of all the $a_{i}$, $B:=[b_{1};\ldots;b_{N}]$ and $U=[u_{1};\ldots;u_{N}]$. Suppose $A<0$. The ISO needs to solve:
\begin{equation}\label{m1}
\max \ [U^{T}AU+B^{T}U]\text{, subject to }1^{T}U=0.
\end{equation}
\noindent where $1$ is the all-one vector of proper size. The solution is:
\begin{equation}\label{lambda}
\lambda^{*N}=\gamma 1^{T}A^{-1}B,
\end{equation}
\begin{equation}\label{ustar}
U^{*N}=\frac{1}{2}A^{-1}\left(\lambda^{*N}\cdot 1-B\right).
\end{equation}

\noindent where $\gamma=\left(trace\left(A^{-1}\right)\right)^{-1}=(1^{T}A^{-1}1)^{-1}$ and index $N$ is used to keep track of the population size. Note also that the optimal social welfare is $\frac{1}{4} \lambda^{2} 1^{T} A^{-1} 1 = \frac{1}{4} \frac{(1^{T} A^{-1} B)^{2}}{1^{T} A^{-1} 1}$.

\begin{theorem}\label{theo2}
	For the SVCG mechanism with quadratic utility functions, if $(a_{i},b_{i})$ satisfy the following:
	\begin{enumerate}
		\item $\underline{a}\le a_{i}\le \bar{a}<0$, $0<\underline{b}\le b_{i}\le \bar{b}$,
		\item $(N-1)H_{max}(N)\le \sum_{i}H_{i}(N),\ F_{total}(N)>0,\ H_{i}(N)>0$, where the argument $N$ denotes that the corresponding quantity refers to the system with agents $2, ... N$, 
	\end{enumerate} 
	then the following holds:
	\begin{enumerate}
		\item There exists a $c^{N}$ satisfying:
		\begin{equation*}
		\frac{(N-1)F_{total}(N)}{\sum_{i}H_{i}(N)}\le c^{N}\le\frac{F_{total}(N)}{H_{max}(N)}.
		\end{equation*}
		Moreover, any such $c^{N}$ satisfies $\lim_{N\to\infty}c^{N}=1$,
		\item $\lim_{N\to\infty}(\lambda^{*N}u_{i}^{*N}-p_{i}^{N})=0$, for all $i$, where the superscript $N$ is used to denote that the corresponding quantity refers to the system with $N$ agents.
	\end{enumerate}
\end{theorem}
\begin{proof}
	It suffices to prove the result for the first agent. Let $A_{-1}=diag(a_{2},\ldots,a_{N})$, $B_{-1}=[b_{2};\ldots;b_{N}]$, $1_{-1}$ be the all-one vector of dimension $N-1$ and $\gamma_{-1}=\left(trace\left(A_{-1}^{-1}\right)\right)^{-1}$. 
		\begin{lemma}\label{lemma1}
		Let $U^{*}$ and $W^{*}$ be the optimal solutions to the problem consisting of all agents, and the problem excluding the first agent, respectively. Then, as the number of agents increases,
		\begin{equation}
		\lim_{N\to\infty} \begin{bmatrix}
		0_{(N-1)\times 1}& I_{N-1}\\
		\end{bmatrix}U^{*}-W^{*}=O(1/N)1_{-1},
		\end{equation} where $0_{(N-1)\times 1}$ is the $N-1$ dimensional column vector of zeroes, and $I_{N-1}$ is the $N-1$ dimensional identity matrix.
		
	\end{lemma}
\begin{proof}
	From \eqref{lambda} and \eqref{ustar}, 
		\begin{align*}
		&\begin{bmatrix}
		0 & I\\
		\end{bmatrix}U^{*}-W^{*}=\\
		&=\frac{1}{2}\begin{bmatrix}
		0 & I\\
		\end{bmatrix}\begin{bmatrix}
		a_{1}^{-1} & 0 \\ 0 & A_{-1}^{-1}
		\end{bmatrix}\left(\gamma\cdot\begin{bmatrix}
		1 & 1_{-1}^{T}\\
		\end{bmatrix}\begin{bmatrix}
		a_{1}^{-1} & 0 \\ 0 & A_{-1}^{-1}
		\end{bmatrix}\begin{bmatrix}
		b_{1}\\ B_{-1}
		\end{bmatrix}\right.\\
		&\left.\begin{bmatrix}
		1\\ 1_{-1}
		\end{bmatrix}-\begin{bmatrix}
		b_{1}\\ B_{-1}
		\end{bmatrix}\right)-\frac{1}{2}A_{-1}^{-1}\left(\gamma_{-1}1^{T}_{-1}A_{-1}^{-1}B_{-1}1_{-1}-B_{-1}\right)\\
		&=\frac{1}{2}\left(\gamma a_{1}^{-1}b_{1}+\left(\gamma-\gamma_{-1}\right)1_{-1}^{T}A_{-1}^{-1}B_{-1}\right)A_{-1}^{-1}1_{-1}.
		\end{align*}
		Since $\underline{a}\le a_{i}\le \bar{a}<0$, $\gamma = \Theta (1/N)$, and $\gamma<0$, $\frac{\bar{a}\underline{b}}{\underline{a}N}\le \gamma a_{1}^{-1}b_{1}\le \frac{\underline{a}\bar{b}}{\bar{a}N}$, $\frac{\bar{a}^{2}}{-\underline{a}N(N-1)}\le \gamma-\gamma_{-1}\le \frac{\underline{a}^{2}}{-\bar{a}N(N-1)}$, $\frac{\underline{a}^{2}\bar{b}}{-\bar{a}^{2}N}\le(\gamma-\gamma_{-1})1^{T}_{-1}A^{-1}_{-1}B_{-1}\le\frac{\bar{a}^{2}\underline{b}}{-\underline{a}^{2}N}$. So, $
		\lim_{N\to\infty}\begin{bmatrix}
		0 & I\\
		\end{bmatrix}U^{*}-W^{*}=0.\qedhere$		 
\end{proof}
Let $\begin{bmatrix}
		0 & I\\
		\end{bmatrix}U^{*}=V^{*}$. From Lemma \ref{lemma1}, $v_{i}^{*}-w_{i}^{*}=O(\frac{1}{N})$
where $v_{i}$ and $w_{i}$ is the $i$-th component of $V^{*}$ and $W^{*}$, respectively. Hence,
\begin{align*}
&\frac{F_{total}}{H_{1}}=\frac{a_{1}u_{1}^{*2}+b_{1}u_{1}^{*}+\sum_{i=2}^{N}(a_{i}v_{i}^{*2}+b_{i}v_{i}^{*})}{\sum_{i=2}^{N}(a_{i}w_{i}^{*2}+b_{i}w_{i}^{*})}\\
&=\frac{a_{1}u_{1}^{*2}+b_{1}u_{1}^{*}+\sum_{i=2}^{N}\left(a_{i}w_{i}^{*2}+b_{i}w_{i}^{*}+G_{1}\right)}{\sum_{i=2}^{N}(a_{i}w_{i}^{*2}+b_{i}w_{i}^{*})},
\end{align*}
\noindent where $G_{1}=(2a_{i}w_{i}^{*}+b_{i})O(\frac{1}{N})+a_{i}O(\frac{1}{N^{2}})$. From equations \eqref{lambda} and \eqref{ustar}, $w_{i}^{*}=\Theta(1)$. Therefore,$\displaystyle\lim_{N\to\infty}\left(F_{total}/H_{1}\right)=1$. Similarly, for all other $i$, $\displaystyle\lim_{N\to\infty}\left(F_{total}/H_{i}\right)=1$. 
Therefore, $\displaystyle\lim_{N\to\infty}\bar{c}^{N}=1$.

Let $H_{min}=\min_{i}H_{i}$. Since $\frac{(N-1)F_{total}}{NH_{max}}\le\underline{c}^{N}\le\frac{(N-1)F_{total}}{NH_{min}}$, $\lim_{N\to\infty}\underline{c}^{N}=1$. Consequently,
$\displaystyle\lim_{N\to\infty}c^{N}=1.$

From Lemma \ref{lemma1}, $W^{*}-V^{*}=-\frac{1}{2}\left(\gamma a_{1}^{-1}b_{1}+\left(\gamma-\gamma_{-1}\right)\xi\right)A^{-1}_{-1}1_{-1}$, where $\xi=1^{T}_{-1}A^{-1}_{-1}B_{-1}$ and $\xi=\Theta(N)$. The payment by Agent $1$ is:
\begin{align*}
p_{1}^{N}&=U_{-1}^{*T}A_{-1}U_{-1}^{*}+B_{-1}^{T}U_{-1}^{*}-V^{*T}A_{-1}V^{*}-B_{-1}^{T}V^{*}\\
&=(U_{-1}^{*}+V^{*})^{T}A_{-1}(U_{-1}^{*}-V^{*})+B^{T}_{-1}(U_{-1}^{*}-V^{*}).
\end{align*}
The difference between Lagrange and VCG payments is:
\begin{align*}
&\lambda^{*N}u_{1}^{*N}-p_{1}^{N}\\
&=\frac{1}{2a_{1}}\gamma(a_{1}^{-1}b_{1}+\xi)\left(\gamma\left(a_{1}^{-1}b_{1}+\xi\right)-b_{1}\right)-p_{1}^{N}\\
&=\frac{1}{2a_{1}}\gamma^{2}(a_{1}^{-1}b_{1}+\xi)^{2}-\frac{b_{1}}{2a_{1}}\gamma(a_{1}^{-1}b_{1}+\xi)\\
&-\bigg[\frac{1}{2}\left[\left(\gamma a_{1}^{-1}b_{1}+\left(\gamma+\gamma_{-1}\right)\xi\right)1_{-1}^{T}-2B_{-1}^{T}\right]A_{-1}^{-1}A_{-1}\\
&+B_{-1}^{T}\bigg]\cdot\frac{-1}{2}\left(\gamma a_{1}^{-1}b_{-1}+\left(\gamma-\gamma_{-1}\right)\xi\right)A_{-1}^{-1}1_{-1}\\
&=\frac{1}{2a_{1}}\gamma^{2}(a_{1}^{-1}b_{1}+\xi)^{2}-\frac{b_{1}}{2a_{1}}\gamma(a_{1}^{-1}b_{1}+\xi)\\
&+\frac{1}{4\gamma_{-1}}\left[\gamma^{2}a_{1}^{-2}b_{1}^{2}+2a_{1}^{-1}b_{1}\gamma^{2}\xi+(\gamma^{2}-\gamma_{-1}^{-2})\xi^{2}\right].
\end{align*}
Since $\gamma=\Theta(\frac{1}{N})$,
\begin{align*}
&\lim_{N\to\infty}\left(\lambda^{*N}u_{1}^{*N}-p_{1}^{N}\right)\\
=&\lim_{N\to\infty}\left[\frac{\gamma^{2}\xi^{2}}{2a_{1}}-\frac{b_{1}\gamma\xi}{2a_{1}}+\frac{b_{1}\gamma^{2}\xi}{2a_{1}\gamma_{-1}}+\frac{(\gamma^{2}-\gamma_{-1}^{2})\xi^{2}}{4\gamma_{-1}}\right]\\
=&\lim_{N\to\infty}\left[\frac{\xi^{2}}{4}(\frac{2\gamma^{2}}{a_{1}}+\frac{\gamma^{2}-\gamma_{-1}^{2}}{\gamma_{-1}})-\frac{b_{1}\gamma\xi}{2a_{1}}(1-\frac{\gamma}{\gamma_{-1}})\right]\\
=&\lim_{N\to\infty}\left[\frac{\xi^{2}}{4}(\frac{\gamma^{2}}{a_{1}}+\gamma-\gamma_{-1})\right].
\end{align*}
By calculation, we have
\begin{align*}
\frac{\gamma^{2}}{a_{1}}+\gamma-\gamma_{-1}=\frac{-1}{a_{1}^2}\left[\frac{1}{(\sum_{i=1}^{N}\frac{1}{a_{i}})^{2}(\sum_{j=2}^{N}\frac{1}{a_{j}})}\right]=O(\frac{1}{N^{3}}).
\end{align*}
Therefore, $\displaystyle\lim_{N\to\infty}\left(\lambda^{*N}u_{1}^{*N}-p_{1}^{N}\right)=0.\qedhere$
\begin{comment}
By calculation, we have
\begin{multline}\label{ON}
\lambda^{*N}u_{1}^{*N}-p_{1}^{N}=\\
\frac{-\prod_{i=2}^{N}a_{i}\left(b_{1}\sum_{i=2}^{N}\prod_{j\ne i, j\ne 1}a_{j}-\sum_{i=2}^{N}b_{i}\prod_{j\ne i, j\ne i}a_{j}\right)^{2}}{4\left(\sum_{i=2}^{N}\prod_{j\ne i, j\ne 1}a_{j}\right)\left(\sum_{i=1}^{N}\prod_{j\ne i}a_{j}\right)^{2}}.
\end{multline}
Since $\underline{a}\le a_{i}\le \bar{a}$, $\underline{b}\le b_{i}\le \bar{b}$, we have,
\begin{equation*}
\frac{\alpha_{1}}{N}\le\lambda^{*N}u_{1}^{*N}-p_{1}^{N}\le \frac{\alpha_{2}}{N}
\end{equation*}
\noindent where $\alpha_{1}$, $\alpha_{2}$ are constants. Therefore,
\begin{equation*}
\lim_{N\to\infty} (\lambda^{*N}u_{1}^{*N}-p_{1}^{N})=0
\end{equation*}
\end{comment}
\end{proof}
\subsection{Dynamic Deterministic Systems}
The VCG scheme can be extended to the important case of deterministic dynamic systems. One can simply consider the entire sequence of actions taken by an agent as a vector action, i.e., as an open-loop control, where the entire decision on the sequence of controls to be employed is taken at the initial time, and so treatable as a static problem.

 For agent $i$, let $F_{i}(x_{i}(t),u_{i}(t))$ be its one-step utility function at time $t$. Suppose that its state evolves as:
\begin{equation}\label{A0}
x_{i}(t+1)=g_{i}(x_{i}(t),u_{i}(t)).
\end{equation}
An example of a wind turbine model can be found in \cite{6713053} where $x(t)=[x_{1}(t),x_{2}(t),x_{3}(t)]^{T}$, $x_{1}(t)$ denotes rotor speed, $x_{2}(t)$ denotes drive train torsion, and $x_{3}(t)$ denotes generator speed. Control $u(t)$ is the collective blade pitch angle and the state equation can be written as $\dot{x}(t)=Ax(t)+Bu(t)$. We assume that $F_{i}=-\infty$ when $(x_{i}(t),u_{i}(t))$ do not satisfy the state dynamic constraint $g_{i}$ or constraints on $(x_{i}(t),u_{i}(t))$ are violated. The ISO asks each agent $i$ to bid its one-step utility functions, state equations and initial condition. Denote the one-step utility function bids made by agent $i$ by $\{\hat{F}_{i}(x_{i}(t),u_{i}(t)), t = 0, 1, \ldots , T-1\}$, its state equation bids by $\{\hat{g}_{i}, t = 0, 1, \ldots , T-1\}$, and its initial condition bid by $\hat{x}_{i,0}$. The ISO then calculates $(x^{*}_{i}(t),u^{*}_{i}(t))$ as the optimal solution, assumed to be unique, to:
\begin{equation*}
\max \ \sum_{i=1}^{N}\sum_{t=0}^{T-1}\hat{F}_{i}(x_{i}(t),u_{i}(t))
\end{equation*}
\noindent subject to
\begin{equation*}
x_{i}(t+1)=\hat{g}_{i}(x_{i}(t),u_{i}(t)),\forall i \text{ and } \forall t,
\end{equation*}
\begin{equation}\label{A1}
\sum_{i=1}^{N}u_{i}(t)=0,\forall t,
\end{equation}
\begin{equation}\label{A2}
x_{i}(0)=\hat{x}_{i,0},\forall i.
\end{equation}
Denote this problem as $(\hat{F},\hat{g},\hat{x}_{0})$. We can extend the VCG payment $p_{i}$ to the deterministic dynamic system. Let
\begin{equation*}
p_{i}:=\sum_{j\ne i}\sum_{t=0}^{T-1}\hat{F}_{j}(x^{(i)}_{j}(t),u^{(i)}_{j}(t))-\sum_{j\ne i}\sum_{t=0}^{T-1}\hat{F}_{j}(x^{*}_{j}(t),u^{*}_{j}(t)).
\end{equation*}
Here $(x^{(i)}_{j}(t),u^{(i)}_{j}(t))$ is the optimal solution to the following problem, which is assumed to be unique:
\begin{equation*}
\max \ \sum_{j\ne i}\sum_{t=0}^{T-1}\hat{F}_{j}(x_{j}(t),u_{j}(t))
\end{equation*}
\noindent subject to
\begin{equation*}
x_{j}(t+1)=\hat{g}_{j}(x_{j}(t),u_{j}(t)),\text{ for }j\ne i \text{ and } \forall t,
\end{equation*}
\begin{equation*}
\sum_{j\ne i}u_{j}(t)=0,\forall t,
\end{equation*}
\begin{equation*}
x_{j}(0)=\hat{x}_{j,0},\text{ for }j\ne i.
\end{equation*}
More generally, we can consider a Groves payment $p_{i}$:
\begin{equation*}
p_{i}:=h_{i}(\boldsymbol{\hat{F}}_{-i})-\sum_{j\ne i}\sum_{t=0}^{T-1}\hat{F}_{j}(x^{*}_{j}(t),u^{*}_{j}(t)),
\end{equation*}
where $h_{i}$ is any arbitrary function. 
\begin{theorem}\label{dtruth}
Truth-telling of utility function, state dynamics and initial condition ($\hat{F}_{i}= F_{i}$, $\hat{g}_{i}= g_{i}$ and $\hat{x}_{i,0}=x_{i,0}$) is a dominant strategy equilibrium under the Groves mechanism for a dynamic system.
\end{theorem}
\begin{proof}
	Let $\hat{F}:=(\hat{F}_{1},...,\hat{F}_{i},...,\hat{F}_{N})$, $\hat{g}:=(\hat{g}_{1}...,\hat{g}_{i},...,\hat{g}_{N})$, and $\hat{x}_{0}:=(\hat{x}_{1,0},...,\hat{x}_{i,0},...,\hat{x}_{N,0})$.
	Suppose agent $i$ announces the true one-step utility function $F_{i}$, true state dynamics $g_{i}$, and true initial condition $x_{i,0}$. Let
	 $\bar{F}:=(\hat{F}_{1},...\hat{F}_{i-1},F_{i},\hat{F}_{i+1},...,\hat{F}_{N})$, $\bar{g}:=(\hat{g}_{1},...\hat{g}_{i-1},g_{i},\hat{g}_{i+1},...,\hat{g}_{N})$, and $\bar{x}_{0}:=(\hat{x}_{1,0},...\hat{x}_{i-1,0},x_{i,0},\hat{x}_{i+1,0},...,\hat{x}_{N,0})$. Let $(\bar{x}^{*}_{i}(t),\bar{u}^{*}_{i}(t))$ be what ISO assigns and $p_{i}(\bar{F},\bar{g},\bar{x}_{0})$ be what ISO charges when $(\bar{F},\bar{g},\bar{x_{0}})$ is announced by agents. Let $(x^{*}_{i}(t),u^{*}_{i}(t))$ be what ISO assigns, and $p_{i}(\hat{F},\hat{g},\hat{x}_{0})$ be what ISO charges when $(\hat{F},\hat{g},\hat{x}_{0})$ is announced by agents. Define $\bar{F}(x_{i}(t),u_{i}(t)):=\sum_{i}\bar{F}_{i}(x_{i}(t),u_{i}(t))$.
	 
	 For agent $i$, the difference between net utility resulting from announcing $(F_{i},g_{i},x_{i,0})$ and $(\hat{F}_{i},\hat{g}_{i},\hat{x}_{i,0})$ is
	 \begin{align*}
	 \Big[\sum_{t}F_{i}(\bar{x}^{*}_{i}(t),\bar{u}^{*}_{i}(t))-p_{i}(\bar{F},\bar{g},\bar{x}_{0})\Big]-\Big[\sum_{t}&F_{i}(x^{*}_{i}(t),u^{*}_{i}(t))\\
	 &-p_{i}(\hat{F},\hat{g},\hat{x}_{0})\Big]
	 \end{align*}
	 \begin{align}\label{dge}
	 &=\sum_{t}F_{i}(\bar{x}^{*}_{i}(t),\bar{u}^{*}_{i}(t))-h_{i,t}(\bar{F}_{-i})+\sum_{j\ne i}\sum_{t}\hat{F}_{j}(\bar{x}^{*}_{i}(t),\bar{u}^{*}_{i}(t)) \nonumber
	 \\
	 &-\sum_{t}F_{i}(x^{*}_{i}(t),u^{*}_{i}(t))+h_{i,t}(\hat{F}_{-i})-\sum_{j\ne i}\sum_{t}\hat{F}_{j}(x^{*}_{i}(t),u^{*}_{i}(t))\nonumber\\
	 &=\sum_{t}\bar{F}(\bar{x}^{*}_{i}(t),\bar{u}^{*}_{i}(t))-\sum_{t}\bar{F}(x^{*}_{i}(t),u^{*}_{i}(t))
	 \end{align}
	 There are two cases: 1. When $(x^{*}_{i}(t),u^{*}_{i}(t))$ also satisfy the state dynamic constraint $g_{i}$, the RHS of \eqref{dge} $\ge 0$ since $(\bar{x}^{*}_{i}(t),\bar{u}^{*}_{i}(t))$ is the optimal solution to the problem $(\bar{F},\bar{g},\bar{x}_{0})$; 2. When $(x^{*}_{i}(t),u^{*}_{i}(t))$ does not satisfy the state dynamics, $F_{i}=-\infty$ and hence the RHS of \eqref{dge} $\ge 0$. 
\end{proof}
Theorem \ref{dtruth} extends to the case of state and input constraints since $F_{i}=-\infty$ outside its feasible set.

Consider now the Scaled VCG mechanism:
\begin{equation*}
p_{i}:=c\sum_{j\ne i}\sum_{t=0}^{T-1}\hat{F}_{j}(x^{(i)}_{j}(t),u^{(i)}_{j}(t))-\sum_{j\ne i}\sum_{t=0}^{T-1}\hat{F}_{j}(x^{*}_{j}(t),u^{*}_{j}(t)).
\end{equation*}
As in the static case, there exists a range of values $c$ that simultaneously achieves IC, EF, BB and IR. 
\begin{comment}
\begin{theorem}\label{DC}
Let $u^{*}(t)$ be the optimal solution to the following problem:
\begin{equation*}
\max\ \sum_{i}\sum_{t}F_{i}(x_{i}(t),u_{i}(t)), \text{ subject to }\eqref{A0}, \eqref{A1}, \eqref{A2},
\end{equation*}
and let $u^{(i)}(t)$ be the optimal solution to the following problem:
\begin{equation*}
\max\ \sum_{j\ne i}\sum_{t}F_{j}(x_{j}(t),u_{j}(t)),
\end{equation*}
\noindent subject to
\begin{equation*}
x_{j}(t+1)=g_{j}(x_{j}(t),u_{j}(t)),\text{ for }j\ne i,
\end{equation*}
\begin{equation*}
\sum_{j\ne i}u_{j}(t)=0, \text{ for all }t \text{ and }x_{j}(0)=x_{j,0}\text{ for }j\ne i.
\end{equation*}

	Let $H_{i}:=\sum_{j\ne i}\sum_{t}F_{j}(x_{j}^{(i)}(t),u_{j}^{(i)}(t))$, and let $H_{max}=\max_{i}H_{i}$. Let $F_{total}:=\sum_{t}\sum_{j}F_{j}(x^{*}_{j}(t),u^{*}_{j}(t))$. If $F_{total}> 0$, $H_{i}>0$ for all $i$, and MPB \eqref{Market power balance} condition holds, then there exists an $\underline{c}$ and $\bar{c}$, with $\underline{c} \leq \bar{c}$ such that if the constant $c$ is chosen in the range $[\underline{c},\bar{c}]$, then the Scaled VCG mechanism for the deterministic dynamic system satisfies IC, EF, BB and IR at the same time.
\end{theorem}

The proof of Theorem \ref{DC} is similar to the proof of Theorem \ref{c} and is omitted here.
As in the static case, we suppose that from experience, the ISO can choose a value of
$c$ in this range, which does not depend on the agents' bids, to achieve BB and IR.
\end{comment}
\subsection{Deterministic Linear Systems with Quadratic Costs}
As in the static case, the SVCG mechanism is asymptotically Lagrange optimal for linear systems with quadratic costs as the number of agents goes to infinity, as shown below. Consider
$F_{i}(x_{i}(t),u_{i}(t))=q_{i}x_{i}^{2}(t)+r_{i}u_{i}^{2}(t)$ and $x_{i}(t+1)=a_{i}x_{i}(t)+b_{i}u_{i}(t)$. Suppose $q_{i} \leq 0$, and $r_{i} < 0$. The following Theorem can be viewed as a generalization of Theorem \ref{theo2} that allows for state dynamics and a horizon larger than 1.
\begin{theorem}\label{theo3}
	For the SVCG mechanism with quadratic utility functions and linear state dynamics, if $(a_{i},b_{i},p_{i},q_{i})$ satisfy
	\begin{enumerate}
		\item $\underline{a}\le |a_{i}|\le \bar{a}$, $\underline{b}\le |b_{i}|\le \bar{b}$, $\underline{q}\le q_{i}\le \bar{q}<0$ and $\underline{r}\le r_{i}\le \bar{r}<0$, 
		\item $(N-1)H_{max}(N)\le\sum_{i}H_{i}(N)$, $F_{total}(N)> 0$ and $H_{i}(N)>0$ for all $i$.
	\end{enumerate}
	Then the following hold:
	\begin{enumerate}
		\item There exist $\underline{c}^{N} \le \bar{c}^{N}$ such that for any $c^{N} \in [ \underline{c}^{N}, \bar{c}^{N}]$,  BB and IR hold. Moreover, $\lim_{N\to\infty}c^{N}=1$,
		\item $\lim_{N\to\infty}\left(\sum_{t}\left(\lambda^{*N}(t)u_{i}^{N}(t)\right)-p_{i}^{N}\right)=0$, for all $i$.
	\end{enumerate}
\end{theorem}
\begin{proof}
	Let $X(t)=(x_{1}(t),x_{2}(t),...,x_{N}(t))^{T}$, $U(t)=(u_{1}(t),u_{2}(t),...,u_{N}(t))^{T}$, $A=diag(a_{1},a_{2},...,a_{N})$, $B=diag(b_{1},b_{2},...,b_{N})$,
	$Q=diag(q_{1},q_{2},...,q_{N})$, and
	$R=diag(r_{1},r_{2},...,r_{N})$. Without loss of generality, assume $x_{j}(0)=0$ for all $j$. The utility maximization problem can be rewritten as the following Linear-Quadratic (LQ) problem:
	\begin{equation}\label{obj}
	\max \ \sum_{t=0}^{T-1}X^{T}(t)QX(t)+U^{T}(t)RU(t)
	\end{equation}
	\noindent subject to
	\begin{equation}\label{state}
	X(t+1)=AX(t)+BU(t),
	\end{equation}
	\begin{equation*}
	1^{T}U(t)=0, \forall t.
	\end{equation*}
	By substituting \eqref{state} into \eqref{obj}, and using the fact that open-loop optimal control is equivalent to the closed-loop optimal solution to LQ problem, we have the following equivalent augmented LQ problem: 
	\begin{equation}\label{d1}
	\max\ [\Omega^{T}(t)W\Omega(t)+V^{T}\Omega(t)]\text{, subject to }Y^{T}\Omega(t)=0.
	\end{equation}
	\noindent where $\Omega:=(U_{1};U_{2};...;U_{N})$, and $U_{i}=(u_{i}(0);u_{i}(1);...;u_{i}(T-1))$, $W$ and $V$ are formed by multiplication and addition of $A, B, Q, R$ and $Y:=[I_{T};I_{T};...;I_{T}]$ with $N$ $T$-dimensional identity matrix $I_{T}$. More specifically, $W$ can be partitioned into diagonal blocks: $W=diag(W_{1},...,W_{N})$, where each block $W_{i}$ is a $T\times T$ square matrix consisting of multiplication and addition of $a_{i}$, $b_{i}$, $q_{i}$, $r_{i}$.
	
	Noting that the optimization problem \eqref{d1} is in the same form as \eqref{m1}, the unique Lagrange multiplier $\bm{\lambda}$ is calculated as $\lambda^{*}=\Gamma Y^{T}W^{-1}V$, where $\Gamma=(Y^{T}W^{-1}Y)^{-1}$. The key to the proof of Theorem \ref{theo2} is to show that $\gamma$ is $\Theta(1/N)$. (Note that $f(N) = \Omega (g(N))$ if $f(N) = \mathcal{O}(g(N))$ as well as $g(N) = \Omega(f(N))$). Similarly, by expanding $\Gamma=(W_{1}^{-1}+W_{2}^{-1}+...+W_{N}^{-1})^{-1}$ and applying bounded inverse theorem \cite{MichaelRenardy2004}, $\vert\vert\Gamma\vert\vert$ is also $\Theta(1/N)$ since $a_{i}$, $b_{i}$, $q_{i}$, $r_{i}$ are all uniformly bounded.
	
	Let $\Omega^{*}$ be the optimal solution to problem \eqref{d1}  consisting of all agents and let $\Psi^{*}$ be the optimal solution to the problem excluding the first agent. By replacing $A$, $B$ and $1$ with $W$, $V$ and $Y$ respectively,
	\begin{equation*}
	\lim_{N\to\infty} \begin{bmatrix}
	0_{(N-1)T\times T}& I_{(N-1)T}\\
	\end{bmatrix}\Omega^{*}-\Psi^{*}=0.
	\end{equation*}
	Let $\begin{bmatrix}
		0 & I\\
		\end{bmatrix}\Omega^{*}=\Phi^{*}$. From above, $\Phi_{i}^{*}-\Psi_{i}^{*}=O(\frac{1}{N})1$
where $\Phi_{i}$ and $\Psi_{i}$ is the $i$-th $T$-length component of $\Phi^{*}$ and $\Psi^{*}$, respectively. Hence,
\begin{align*}
&\frac{F_{total}}{H_{1}}=\frac{U_{1}^{*T}W_{1}U_{1}^{*}+V_{1}^{T}U_{1}^{*}+\sum_{i=2}^{N}(\Phi_{i}^{*T}W_{i}\Phi^{*}_{i}+V^{T}_{i}\Phi_{i}^{*})}{\sum_{i=2}^{N}(\Psi_{i}^{*T}W_{i}\Psi^{*}_{i}+V^{T}_{i}\Psi_{i}^{*})}=\\
&\frac{U_{1}^{*T}W_{1}U_{1}^{*}+V_{1}^{T}U_{1}^{*}+\sum_{i=2}^{N}(\Psi_{i}^{*T}W_{i}\Psi^{*}_{i}+V^{T}_{i}\Psi_{i}^{*}+G_{1})}{\sum_{i=2}^{N}(\Psi_{i}^{*T}W_{i}\Psi^{*}_{i}+V^{T}_{i}\Psi_{i}^{*})}
\end{align*}
\noindent where $G_{1}=(2\Psi_{i}^{*T}W_{i}1+V_{i}^{T}1)O(\frac{1}{N})+1^{T}W_{i}1\cdot O(\frac{1}{N^{2}})$. Since $\Psi_{i}^{*}=\Theta(1)1$, we have $\displaystyle\lim_{N\to\infty}\left(F_{total}^{N}/H_{1}^{N}\right)=1$. Similarly, for all other $i$, $\displaystyle\lim_{N\to\infty}\left(F_{total}^{N}/H_{i}^{N}\right)=1$.
Therefore, $\displaystyle\lim_{N\to\infty}\bar{c}^{N}=1$.
	
	Let $H_{min}=\min_{i}H_{i}$. Since $\frac{(N-1)F_{total}}{NH_{max}}\le\underline{c}^{N}\le\frac{(N-1)F_{total}}{NH_{min}}$, $\lim_{N\to\infty}\underline{c}^{N}=1$.
	Consequently, $\displaystyle\lim_{N\to\infty}c^{N}=1.$
		
	From Lemma \ref{lemma1},
	\begin{align*}
	\Psi^{*}-\Phi^{*}=\frac{-1}{2}W_{-1}^{-1}Y_{-1}(\Gamma W_{-1}^{-1}V_{-1}+(\Gamma-\Gamma_{-1})\Xi),
	\end{align*}
	\noindent where $W_{-1}$, $V_{-1}$ are formed by removing $W_{1}$ and $V_{1}$ from $W$ and $V$, respectively. $Y_{-1}=[I_{T};...;I_{T}]$ with $(N-1)$ $T$-dimensional identity matrix. $\Xi=Y_{-1}^{-1}W_{-1}^{-1}V_{-1}$ and $\Xi=O(N)1$. Similarly as in Theorem \ref{theo2}, 
	\begin{align*}
	&\lim_{N\to\infty} \left(\lambda^{*T}U_{1}^{*}-p_{1}^{N}\right)\\
	&=\lim_{N\to\infty}\Bigg[\frac{1}{2}(V_{1}^{T}W_{1}^{-1}+\Xi^{T})\Gamma^{T}\big[W_{1}^{-1}\Gamma\big(W_{1}^{-1}V_{1}+\Xi\big)\\
	&-W_{1}^{-1}V_{1}\big]-\bigg[\frac{1}{2}\Big[\left(\Gamma W_{-1}^{-1}V_{-1}+(\Gamma+\Gamma_{-1})\Xi\right)^{T}Y_{-1}^{T}-2V_{-1}\Big]^{T}\\
	&W_{-1}^{-1}W_{-1}+V_{-1}^{T}\bigg]\cdot\frac{-1}{2}W_{-1}^{-1}Y_{-1}(\Gamma W_{-1}^{-1}V_{-1}+(\Gamma -\Gamma_{-1})\Xi)\Bigg]\\
	&=\lim_{N\to\infty}\bigg(\frac{1}{4}\Xi^{T}\Big(\Gamma^{T}W^{-1}_{1}\Gamma+\Gamma-\Gamma_{-1}\Big)\Xi\bigg)
	\end{align*}
	It is straightforward to see that,
	\begin{align*}
	&\Gamma^{T}W^{-1}_{1}\Gamma+\Gamma-\Gamma_{-1}\\
	=&-\left(\sum_{i=1}^{N}W_{i}^{-1}\right)^{-1}W_{1}^{-1}\left(\sum_{i=1}^{N}W_{i}^{-1}\right)^{-1}W_{1}^{-1}\left(\sum_{i=2}^{N}W_{i}^{-1}\right)^{-1}\\
	=&O(\frac{1}{N^{3}}).
	\end{align*}
	Consequently, $\displaystyle\lim_{N\to\infty} \left(\lambda^{*T}U_{i}^{*}-p_{i}^{N}\right)=0. \qedhere$
\end{proof}

\section{Dynamic Stochastic VCG}\label{SCDC}
In the previous section, the VCG mechanism was naturally extended to deterministic dynamic systems by employing an open-loop solution. A new complication arises when agents are stochastic dynamic systems. The states of agents evolve randomly, and so we need to consider closed-loop control laws for each agent. Such closed-loop control laws depend on the observations of the agents, which are generally private random variables. Hence the problem therefore arises of ensuring that each agent reveals its ``true" observation \emph{at each and every time instant}. (Other unknowns such as system dynamic equations, noise statistics, and utility functions, can be considered as part of the first observation). 

However, a fundamental difficulty arises with respect to ensuring social welfare optimality of the stochastic dynamic system. Since an agent's intertemporal payoff depends on the future payments and allocations in a dynamic game, the agent's current bid need not maximize its current payoff. What's more, since dishonest bids distort current and future allocations in different ways, an agent's optimal bid will depend on others' bids.

To see this, it is sufficient to consider the case where all system parameters -- system dynamics, noise statistics, and utility functions -- are known to all agents, and where each agent can completely observe its own private state $x_{i}(t)$, with the only complication being that it cannot observe the states of other agents. For agent $i$, let $w_{i}(t)$ be the discrete-time noise process affecting state $x_{i}(t)$ via the state evolution equation:
	\begin{equation*}
	x_{i}(t+1)=g_{i}(x_{i}(t), u_{i}(t), w_{i}(t)),
	\end{equation*}
	where $x_{i}(0)$ is independent of $w_{i}$. An LQG model for wind turbine control can be found in \cite{6713053}, where $w(t)$ is the turbine system noise. The uncertainties of all the agents are independent. The ISO aims to maximize the social welfare:
	
	\begin{equation*}
	\max \ \mathbb{E}\sum_{i=1}^{N}\sum_{t=0}^{T-1}F_{i}(x_{i}(t),u_{i}(t))
	\end{equation*}
	\noindent subject to
	\begin{equation}
	x_{i}(t+1)=g_{i}(x_{i}(t), u_{i}(t), w_{i}(t)),\forall t.
	\end{equation}
	\begin{equation}
	\sum_{i=1}^{N}u_{i}(t)=0,\text{ for }\forall t.
	\end{equation}
We will assume that $F_{i}$, $g_{i}$ and the distributions of the uncertainties are known to the ISO. We comment in the sequel on the further difficulty that arises when they are unknown. We focus on the issue of truth-telling by the agents of their states. 

Suppose that agents bid their states $x_{i}(t)$ as $\hat{x}_{i}(t)$. 
	A straightforward extension of the static Groves mechanism, which we will see does not work, would be to collect a payment $p_{i}(t)$ at time $t$ from agent $i$, defined as
	\begin{align*}
	&p_{i}(t)\\
	=&h_{i}(\hat{X}_{-i}(t))-\mathbb{E}\sum_{j\ne i}\sum_{\tau=t}^{T-1}\left[F_{j}(x_{j,t}(\tau),u_{j,t}^{*}(\tau))\mid X(t)=\hat{X}(t)\right],
	\end{align*}
	\noindent where $\hat{x}_{i}(t)$ is what agent $i$ bids for his state at time $t$, $\hat{X}_{-i}(t)=[\hat{x}_{1}(t),...,\hat{x}_{i-1}(t),\hat{x}_{i+1}(t),...,\hat{x}_{N}(t)]^{T}$, $u_{j,t}^{*}(\tau)$ is the optimal solution to the following problem:
	\begin{equation*}
	\max\ \mathbb{E}\sum_{j=1}^{N}\sum_{\tau=t}^{T-1}\left[F_{j}(x_{j}(\tau),u_{j}(\tau))\mid X(t)=\hat{X}(t)\right]
	\end{equation*}
	\noindent subject to
	\begin{equation*}
	x_{j}(\tau+1)=g_{j}(x_{j}(\tau),u_{j}(\tau),w_{j}(\tau)),
	\end{equation*}
	\begin{equation}
	\sum_{j=1}^{N}u_{j}(\tau)=0,\text{ for }t\le\tau\le T-1,
	\end{equation}
	\begin{equation*}
	\hat{X}(t)=[\hat{x}_{1}(t),...,\hat{x}_{N}(t)]^{T},
	\end{equation*}
	and $x_{j,t}(\tau+1)=g_{j}(x_{j,t}(\tau),u_{j,t}(\tau),w_{j,t}(\tau))$.
	It is easy to verify that truth-telling of states by all agents forms a subgame perfect Nash equilibrium since truth-telling of $x_{i}(t)$ for agent $i$ is a best response given that all other agents bid truthfully for all $\tau \geq t$. 
	
	However, truth-telling of states does \emph{not} constitute a dominant strategy because another agent $j$ may bid $\hat{x}_{j}(t+1)$ at time $t+1$ truthfully, but lie about the state $x_j(t)$ at time $t$ in order to obtain a preferable state at the next time $t+1$. More specifically, if we assume all agents will bid truthfully from $t+1$ onward, then at time $t$, if agent $j$ bids some untruthful $\hat{x}_{j}(t)$, truth-telling of state for agent $i$ will be an optimal strategy only if agent $j$ continues to bid ``an untruthful but consistent'' $\hat{x}_{j}(t)$ which stems from his untruthful bid $\hat{x}_{j}(t)$. By ``consistent" we mean the state that would result from the untruthful $\hat{x}_{j}(t)$ but with the truthful state noise $w_j(t)$. In other words, agent $i$'s will bid truthfully only if agent $j$ ``consistently'' lies about his state, which is not guaranteed using the above payment scheme. 
	
	This additional complication precludes a dominant strategy solution for general stochastic dynamic systems even in the \emph{completely observed case}, and even when all \emph{system parameters} (the system dynamic equations, the noise statistics, and the utility functions) are known to all agents. We conjecture that there does not exist a dominant strategy for each agent that ensures social welfare optimality even in this special context, when agents are general stochastic dynamic systems. All that one can possibly hope for in general is a subgame perfect Nash equilibrium where truth telling by an agent is optimal when all other agents are telling the truth.
	
	We show in the sequel that there is one important exception: LQG agents.\section{Linear Quadratic Gaussian Systems}\label{LQGS}
	We now show that while an incentive compatible strategy presents fundamental challenges for general stochastic dynamic systems even when system parameters are known to all, as noted above, there is a solution for LQG systems in both the \emph{completely observed case} where each agent observes its own private state, and in the \emph{partially observed case} where each agent observes a linear transformation of its state corrupted by white Gaussian noise. In fact, we show that one can then even obtain the following stronger property:
\begin{defn}
We say that a mechanism attains a \emph{subgame perfect dominant strategy equilibrium} if, at every time $t$, truth-telling by each agent of its remaining future private observations is optimal with respect to the conditional expectation of its remaining net utility, irrespective of the strategies of other agents in the future.
\end{defn}

We investigate the structure of LQG systems more carefully. We begin by considering the completely observed case. For agent $i$, let $w_{i}(t)\sim\mathcal{N}(0,\sigma_{i})$ be the discrete-time additive Gaussian white noise process affecting state $x_{i}(t)$ via:
	\begin{equation*}
	x_{i}(t+1)=a_{i}x_{i}(t)+b_{i}u_{i}(t)+w_{i}(t),
	\end{equation*}
	where $x_{i}(0)\sim\mathcal{N}(0,\zeta_{i})$ and is independent of $w_{i}$.
	Each agent has a one-step utility function
	\begin{equation*}
	F_{i}(x_{i}(t),u_{i}(t))=q_{i}x_{i}^{2}(t)+r_{i}u_{i}^{2}(t).
	\end{equation*} 
We suppose that $q_{i} \leq 0$ and $r_{i} < 0$. Let $X(t)=[x_{1}(t),...,x_{N}(t)]^{T}$, $U(t)=[u_{1}(t),...,u_{N}(t)]^{T}$ and $W(t)=[w_{1}(t),...,w_{N}(t)]^{T}$. Let $Q=diag(q_{1},...,q_{N})\leq 0$, $R=diag(r_{1},...,r_{N})<0$, $A=diag(a_{1},...,a_{N})$, $B=diag(b_{1},...,b_{N})$, $\Sigma=diag(\sigma_{1},...,\sigma_{N})>0$ and $Z=diag(\zeta_{1},...,\zeta_{N})>0$. We assume that the ISO knows the true system parameters $A$, $B$, $Q$ and $R$.

Let $RSW:=\sum_{t=0}^{T-1}[X^{T}(t)QX(t)+U^{T}(t)RU(t)]$ be the \emph{random} social welfare, i.e., the variable whose expectation is the social welfare of the agents, and let $SW:=\mathbb{E}[RSW]$ denote the (\emph{expected}) social welfare. $RSW$ could also be called the ``ex-post social welfare'', while $SW$ could be called the ``ex-ante social welfare." 

The ISO aims to maximize the social welfare:

	\begin{equation*}
	\max \ \mathbb{E}\sum_{t=0}^{T-1}\left[X^{T}(t)QX(t)+U^{T}(t)RU(t)\right]
	\end{equation*}
	\noindent subject to
	\begin{equation*}
	X(t+1)=AX(t)+BU(t)+W(t),
	\end{equation*}
	\begin{equation}\label{ubalance}
	1^{T}U(t)=0,\forall t,
	\end{equation}
	\begin{equation*}
	X(0)\sim\mathcal{N}(0,Z), W\sim\mathcal{N}(0,\Sigma).
	\end{equation*}

The key to obtaining subgame perfect dominance is to introduce a ``layered'' payment structure which ensures incentive compatibility for LQG systems. We begin by rewriting the random social welfare, and thereby also the social welfare, in terms more convenient for us. We will decompose  the state $X(t)$ of the entire system comprised of all agents as:
	\begin{equation}\label{xdecompose}
	X(t):=\sum_{s=0}^{t}X(s,t),\ 0\le t\le T-1,
	\end{equation}
	where $X(s,s) := W(s-1)$ for $s\geq 1$ and $X(0,0):=X(0)$. Let
	\begin{equation}\label{projection}
	X(s,t) := AX(s,t-1)+BU(s,t-1), \text{ } 0 \leq s \leq t-1,
	\end{equation}
	with $U(s,t)$ yet to be specified. We suppose that $U(t)$ can also be decomposed as:
	\begin{equation}\label{udecompose}
	U(t):=\sum_{s=0}^{t}U(s,t),\ 0\le t \le T-1.
	\end{equation}
	Then \emph{regardless of how the $U(s,t)$'s are chosen}, as long as the $U(s,t)$'s for $0 \leq s \leq t$ are indeed a decomposition of $U(t)$, i.e., \eqref{udecompose} is satisfied, the random social welfare can be written in terms of $X(s,t)$'s and $U(s,t)$'s as:	
	\begin{equation*}
	RSW=\sum_{s=0}^{T-1}L_{s},
	\end{equation*}
	where $L_{s}$ for $s\ge 1$ is defined as:
	\begin{align}\label{ls}
	L_{s}:&=\sum_{t=s}^{T-1}\bigg[X^{T}(s,t)QX(s,t)+U^{T}(s,t)RU(s,t)\\\nonumber+&2\left(\sum_{\tau=0}^{s-1}X(\tau,t)\right)QX(s,t)+2\left(\sum_{\tau=0}^{s-1}U(\tau,t)\right)RU(s,t)\bigg],
	\end{align}
	and $L_{0}$ is defined as:
	\begin{equation*}
	L_{0}:=\sum_{t=0}^{T-1}\Big[X^{T}(0,t)QX(0,t)+U^{T}(0,t)RU(0,t)\Big].
	\end{equation*}
	Hence,
	$SW=\mathbb{E}\sum_{s=0}^{T-1}L_{s}.$
		
	In the scheme to follow, the ISO will choose all $U(s,t)$'s for future $t$'s at time $s$, based on the information it has at time $s$.  Hence $X(s,t)$ is completely determined by $W(s-1)$, and $U(s,t)$ for $s \leq t \leq T-1$. Indeed $X(s,t)$ can be regarded as the contribution to $X(t)$ of these variables.

	We now define the \emph{LQG ISO Mechanism}. Instead of asking agents to bid their state, we will consider a scheme where agents will be asked to bid their \emph{state noises}. At each stage $s$, the ISO asks each agent $i$ to bid its $x_{i}(s,s)$, defined as equal to $w_{i}(s-1)$. Let $\hat{x}_{i}(s,s)$ be what the agent actually bids, since it may not tell the truth. Based on their bids \{$\hat{x}_i(s,s)$ for $1 \leq i \leq N$\}, the ISO solves the following problem:		
	\begin{equation*}
	\max \ L_{s}
	\end{equation*} 
	\noindent for the system
	\begin{equation*}
	\hat{X}(s,t)=A\hat{X}(s,t-1)+BU(s,t-1),\text{ for }t>s,
	\end{equation*}
	with
	\begin{equation*}
	\hat{X}(s,s)=[\hat{x}_{1}(s,s),...,\hat{x}_{N}(s,s)]^{T},
	\end{equation*}
	subject to the constraint
	\begin{equation*}
	1^{T}U(s,t)=0,\text{ for } s\le t \le T-1.
	\end{equation*}
Here $\hat{X}(s,t)$ is the zero-noise state variable updates starting from the ``initial condition'' $\hat{X}(s,s)$. Let $U^{*}(s,t)$ denote the optimal solution.
	
	The interpretation is the following. Based on the bids, $\hat{X}(s,s)$, which is supposedly a bid of $W(s-1)$, the ISO calculates the trajectory of the linear systems from time $s$ onward, assuming zero state noise from that point on. It then allocates consumptions/generations $U(s,t)$ for future periods $t$ for the corresponding deterministic linear system, with balance of consumption and production \eqref{ubalance} at each time $t$. These can be regarded as generation/consumption allocations taking into account the consequences of the disturbance occurring at time $s$. Thereby we are decomposing the behavior of the system into separate effects caused by the state noise random variables occurring at different times.
	
	Next, the ISO collects a payment $p_{i}(s)$ from agent $i$ at time $s$ as:
	\begin{align*}
	p_{i}(s):=h_{i}(\hat{X}_{-i}(s,s))-\sum_{j\ne i}\sum_{t=s}^{T-1}\bigg[q_{j}\hat{x}_{j}^{2}(s,t)+r_{j}u_{j}^{*2}(s,t)\\+2q_{j}\left(\sum_{\tau=0}^{s-1}\hat{x}_{j}(\tau,t)\right)\hat{x}_{j}(s,t)+2r_{j}\left(\sum_{\tau=0}^{s-1}u_{j}(\tau,t)\right)u_{j}^{*}(s,t)\bigg],
	\end{align*}
	where $\hat{X}_{-i}(s,s)=[\hat{x}_{1}(s,s),...,\hat{x}_{i-1}(s,s),\hat{x}_{i+1}(s,s),...\\
	,\hat{x}_{N}(s,s)]^{T}$, and $h_{i}$ is any arbitrary function (as in the Groves mechanism).
	
	Before proving incentive compatibility, we need to define what is meant by ``rationality" of an agent in a dynamic system where each agent has to take actions at different times. 
	\begin{defn}
	Rational Agents: We say agent $i$ is \emph{rational at time} $T-1$, if it adopts a dominant strategy whenever there exists a unique dominant strategy. An agent $i$ is \emph{rational at time} $t$ if it adopts a dominant strategy at time $t$ under the assumption that all agents including itself are rational at times $t+1, t+2, ..., T-1$, whenever there is a unique such dominant strategy at time $t$.
	\end{defn}
	 A critical property of the LQ problem is that the optimal feedback gain does not depend on the state. A key result needed to show dominance of truth telling is to achieve ``intertemporal" decoupling.  In the case of the quadratic cost, this will be achieved by showing the following Lemma in the sequel: 
\begin{lemma}\label{decoupling}
(Intertemporal Decoupling for LQG Agents) Let $\mathcal{H}_{k}$ be the history up to time $k$, $\mathcal{H}_{k}:=\{w_{i}(t): 1\le i\le N, 0\le t\le k\}$. Then, under the LQG ISO Mechanism,  if system parameters $Q\le 0$, $R<0$, $A$ and $B$ are known, and agents are rational, then for $k<t\le T-1$,
\[
\mathbb{E}[X(k,t)|\mathcal{H}_{k-1}]=0, \text{ and }\mathbb{E}[U(k,t)|\mathcal{H}_{k-1}]=0.
\]
\end{lemma}
The proof of this is provided as part of the following overarching result.
	
	\begin{theorem}\label{stochasticvcg}
		Truth-telling of state $\hat{x}_{i}(s,s)$ for $0\le s \le T-1$, i.e., bidding $\hat{x}_i(s,s) = w_i(s-1)$, is the unique dominant strategy for each agent $i$ under the LQG ISO Mechanism, if system parameters $Q\le 0$, $R<0$, $A$ and $B$ are known, and agents are rational. The LQG ISO Mechanism achieves social welfare optimization.
	\end{theorem}
	\begin{proof}
	We show the result by backward induction. Let Agent $j$, $j\ne i$ bid $\hat{x}_{j}(s,s)$ at time $s$. Given the bids $\hat{x}_j(s,s)$ of other agents, let $J_{i}(s)$ be the net utility of agent $i$ from time $s$ onward if it bids truthful $x_{i}(s,s)$, i.e., $w_{i}(s-1)$, and let $\hat{J}_{i}(s)$ be the net utility if it bids possibly untruthful $\hat{x}_{i}(s,s)$. Let $U^{*}(s,t)$ be the ISO's assignments if agent $i$ bids truthfully and let $\hat{U}^{*}(s,t)$ be the ISO's assignments if agent $i$ bids untruthfully. 
		
		We first consider time $T-1$, since we are employing backward induction. Suppose that $x_i(s,T-1)$ for $0 \leq s \leq T-2$ were the past bids, and $u_i(s,T-1)$ for $0 \leq s \leq T-2$,
were those portions of the allocations for the future already decided in the past. Our interest is on analyzing what should be the current bid $x_i(T-1,T-1)$, and the consequent additional allocation $u_i(T-1,T-1)$. Now
	$x_{i}(s,T-1)$ for $0\le s\le T-2$ depend only on previous bids $x_{i}(s,s)$, and thus those terms can be treated as constants. In addition, the $h_{i}$ term depends only on other agents' bids. As a consequence, when comparing $J_{i}(T-1)$ with $\hat{J}_{i}(T-1)$, one can just regard $h_{i} \equiv 0$. Hence
	    \begin{align*}
	    &J_{i}(T-1)=q_{i}x_{i}^{2}(T-1,T-1)+r_{i}u_{i}^{*2}(T-1,T-1)\\&+2q_{i}\left(\sum_{s=0}^{T-2}x_{i}(s,T-1)\right)x_{i}(T-1,T-1)\\&+2r_{i}\left(\sum_{s=0}^{T-2}u_{i}(s,T-1)\right)u_{i}^{*}(T-1,T-1)\\&+\sum_{j\ne i}\bigg[q_{j}\hat{x}_{j}^{2}(T-1,T-1)+r_{j}u_{j}^{*2}(T-1,T-1)\\&+2q_{j}\left(\sum_{\tau=0}^{T-2}\hat{x}_{j}(\tau,T-1)\right)\hat{x}_{j}(T-1,T-1)\\&+2r_{j}\left(\sum_{\tau=0}^{T-2}u_{j}(\tau,T-1)\right)u_{j}^{*}(T-1,T-1)\bigg].
	    \end{align*}
	    It is seen that $J_{i}(T-1)$ is of the same form as $L_{T-1}$. $\hat{J}_{i}(T-1)$ is obtained by replacing $u_{i}^{*}$ with $\hat{u}_{i}^{*}$. We conclude that $J_{i}(T-1)\ge\hat{J}_{i}(T-1)$ because $u_{i}^{*}$ is the optimal solution to $L_{T-1}$ when $\hat{x}_{i}(T-1,T-1)=x_{i}(T-1,T-1)$. Moreover truth telling is the unique optimal strategy since it is a finite-horizon, discrete-time LQR problem.
	    
	    We next employ induction, and so assume that truth-telling of states is the unique subgame perfect dominant strategy equilibrium at time $k$. Let $\mathcal{H}_{t}$ be the history up to time $t$. If agents are rational, we can take the expectation over future $X(s,s)$, $s\ge k$, which are i.i.d. Gaussian noise vectors, and calculate $J_{i}(k-1)$ (where, as before, we simply take the first Groves term $h_{i} \equiv 0$):
	    \begin{align}\label{star}
	    \begin{split}
	    &J_{i}(k-1)=\\
	    &q_{i}x_{i}^{2}(k-1)+r_{i}u_{i}^{2}(k-1)-p_{i}(k-1)+\mathbb{E}\left[J_{i}(k)|\mathcal{H}_{k-1}\right]\\
	    &=q_{i}\left[x_{i}(k-1,k-1)+\sum_{s=0}^{k-2}x_{i}(s,k-1)\right]^{2}\\
	    &+r_{i}\left[u_{i}(k-1,k-1)+\sum_{s=0}^{k-2}u_{i}(s,k-1)\right]^{2}-p_{i}(k-1)\\
	    &+\mathbb{E}\left[\sum_{t=k}^{T-1}\left(q_{i}x_{i}^{2}(t)+r_{i}u_{i}^{2}(t)-p_{i}(t)\right)\bigg|\mathcal{H}_{k-1}\right].
	    \end{split}
	    \end{align}
	    We now prove Lemma \ref{decoupling}. We first show that $\mathbb{E}[U^{*}(k,k)|\mathcal{H}_{k-1}]=0$. By completing the square for $L_{k}$ in \eqref{ls}, we have the following equivalent problem for the ISO to solve for the $k$-th layer:
	    \begin{multline}\label{completionofsquare}
	    \max\ \sum_{t=k}^{T-1}\Bigg[\left(X(k,t)+\sum_{\tau=0}^{k-1}X(\tau,t)\right)^{T}Q\cdot\\
	    \left(X(k,t)+\sum_{\tau=0}^{k-1}X(\tau,t)\right)
	    +\left(U(k,t)+\sum_{\tau=0}^{k-1}U(\tau,t)\right)^{T}R\cdot\\
	    \left(U(k,t)+\sum_{\tau=0}^{k-1}U(\tau,t)\right)\Bigg].
	    \end{multline}
	    Now, for the fixed $k$ of interest, letting $Y(t) := X(k,t)+\sum_{\tau=0}^{k-1} X(\tau,t)$, and $V(t) := U(k,t) + \sum_{\tau=0}^{k-1} U(\tau,t)$, we have 
$Y(t) = AY(t-1)+BV(t-1)$ for $t \geq k+1$. The ``initial" condition is $Y(k) = X(k)$. For this linear system, the optimal control law for cost \eqref{completionofsquare} under the balancing constraint for all $t$ is linear in the state \cite{8274944}. Denoting the optimal gain by $K(t)$ (whose calculation can be found in \cite{Kumar1986}),
	    \begin{equation*}
	    U^{*}(k,k)+\sum_{\tau=0}^{k-1}U(\tau,k)=K(k)\left[X(k,k)+\sum_{\tau=0}^{k-1}X(\tau,k)\right].
	    \end{equation*}
	    Similarly, at time $k-1$, the ISO chooses the allocation at time $k$ by using the same gain $K(t)$ applied to that portion of the state at time $k$ resulting from disturbances up to time $k-1$:
	    \begin{align*}
	    U(k-1,k)+\sum_{\tau=0}^{k-2}U(\tau,k)=\sum_{\tau=0}^{k-1}U(\tau,k)=K(k)\cdot\\
	    \left[X(k-1,k)+\sum_{\tau=0}^{k-2}X(\tau,k)\right]=K(k)\left[\sum_{\tau=0}^{k-1}X(\tau,k)\right].
	    \end{align*}
	    Consequently,
	    \begin{equation*}
	    \mathbb{E}[U^{*}(k,k)|\mathcal{H}_{k-1}]=K(k)\mathbb{E}[X(k,k)|\mathcal{H}_{k-1}]=0,
	    \end{equation*}
	    since all agents are truth-telling at time $k$, i.e., $\mathbb{E}[X(k,k)|\mathcal{H}_{k-1}]=\mathbb{E}[W(k-1)]=0$. From \eqref{projection}, by linearity of the system, $\mathbb{E}[X(k,t)|\mathcal{H}_{k-1}]=0$, $k<t\le T-1$, and $\mathbb{E}[U(k,t)|\mathcal{H}_{k-1}]=0$, $k<t\le T-1$.
	    Therefore, for $k\le t\le T-1$,
	    \begin{align*}
	    &\mathbb{E}[x_{i}^{2}(t)|\mathcal{H}_{k-1}]=\mathbb{E}\left[\sum_{\tau=k}^{t}x_{i}(\tau,t)+\sum_{s=0}^{k-1}x_{i}(s,t)\right]^{2}\\
	    &=\left[\sum_{s=0}^{k-1}x_{i}(s,t)\right]^{2}+C=x_{i}^{2}(k-1,t)+2x_{i}(k-1,t)\cdot\\
	    &\sum_{s=0}^{k-2}x_{i}(s,t)+\left[\sum_{s=0}^{k-2}x_{i}(s,t)\right]^{2}+C,
	    \end{align*}
	    where $C$ is a fixed term corresponding to the variance of $\sum_{\tau=k}^t x_{i}(\tau, t)$ and $\left[\sum_{s=0}^{k-2}x_{i}(s,t)\right]^{2}$ can be treated as a constant since it depends only previous bids. Similarly, for $t \geq k$, 
	    \begin{align*}
	    &\mathbb{E}[u_{i}^{2}(t)|\mathcal{H}_{k-1}]=\\
	    &u_{i}^{2}(k-1,t)+2u_{i}(k-1,t)\sum_{s=0}^{k-2}u_{i}(s,t)+\left[\sum_{s=0}^{k-2}u_{i}(s,t)\right]^{2}+C,
	    \end{align*}
	    We also have,
	    \begin{equation*}
	    \mathbb{E}[p_{i}(t)|\mathcal{H}_{k-1}]=const.,
	    \end{equation*}
	    since $\mathbb{E}[x_{j}(t,\tau)|\mathcal{H}_{k-1}]=0$ and $\mathbb{E}[u_{j}(t,\tau)|\mathcal{H}_{k-1}]=0$, for $\tau\ge t$. By ignoring the constant term,
	    \begin{align*}
	    &J_{i}(k-1)=\\
	     &q_{i}x_{i}^{2}(k-1,k-1)+2q_{i}x_{i}(k-1,k-1)\sum_{s=0}^{k-2}x_{i}(s,k-1)\\
	    +&r_{i}u_{i}^{2}(k-1,k-1)+2r_{i}u_{i}(k-1,k-1)\sum_{s=0}^{k-2}u_{i}(s,k-1)\\
	    +&\sum_{t=k}^{T-1}\left[q_{i}x_{i}^{2}(k-1,t)+2q_{i}x_{i}(k-1,t)\sum_{s=0}^{k-2}x_{i}(s,t)\right]\\
	    +&\sum_{t=k}^{T-1}\left[r_{i}u_{i}^{2}(k-1,t)+2r_{i}u_{i}(k-1,t)\sum_{s=0}^{k-2}u_{i}(s,t)\right]\\
	    -&p_{i}(k-1)\\
	    &=\sum_{t=k-1}^{T-1}\Bigg[q_{i}x_{i}^{2}(k-1,t)+r_{i}u_{i}^{2}(k-1,t)\\
	    &+2q_{i}\left(\sum_{\tau=0}^{k-2}x_{i}(\tau,t)\right)x_{i}(k-1,t)\\
	    &+2r_{i}\left(\sum_{\tau=0}^{k-2}r_{i}(\tau,t)\right)r_{i}(k-1,t)\Bigg]-p_{i}(k-1).
	    \end{align*} 
	    It is straightforward to check that $J_{i}(k-1)$ is of the same form as $L_{k-1}$ and thus we conclude that truth-telling $\hat{x}_{i}(k-1,k-1)=x_{i}(k-1,k-1)$ is the unique dominant strategy for agent $i$ at time $k-1$.
\end{proof}
In the proof we have actually established the following stronger result:
\begin{cor}
In the stochastic VCG mechanism, truth-telling of states constitutes a subgame perfect dominant strategy equilibrium. 
\end{cor}
\subsection{LQG systems with unknown system parameters}
The case where the ISO does not know the system parameters, system dynamic equations, noise statistics, and utility functions, poses formidable difficulties and we conjecture that there is no mechanism with truth telling as a dominant strategy. Above, the key to proving incentive compatibility for the layered VCG mechanism lies in the fact that the optimal feedback gain $K(k)$ remains unchanged for each round of bids. This is due to the fact that $K(k)$ is only a function of $Q$, $R$, $A$, and $B$. Therefore, if bidding of system parameters at the beginning is allowed, then the layered VCG mechanism is not incentive compatible. We show this by the following counterexample.
\begin{exmp}\label{bidabqr}
Let $T=4$. The agents' system equations and cost matrices have the following parameters: $(a_{1},a_{2},a_{3},a_{4})=(1,1,1,1)$, $(b_{1},b_{2},b_{3},b_{4})=(1,1,1,1)$, $(q_{1},q_{2},q_{3},q_{4})=(-1,-1,-1,-1)$, $(r_{1},r_{2},r_{3},r_{4})=(-1,-1.1,-1.2,-1.1)$, $(\zeta_{1},\zeta_{2},\zeta_{3},\zeta_{4})=(0.3,0.32,0.31,0.3)$ and $(\sigma_{1},\sigma_{2},\sigma_{3},\sigma_{4})=(0.1,0.11,0.11,0.12)$. If system operator knows all the parameters of agents, and every agent bid its true state, then the expected net utility of agent $1$ (expected total utility minus expected total payment) is $0.629$. When agents are also allowed to bid their system parameters at the beginning, truth-telling of state may not be incentive compatible. Suppose that agents $2$, $3$, $4$ remain truthful, namely, bid their true system parameters at the beginning and their true states at all times. Suppose now that agent $1$ intentionally bids an untruthful $\hat{q}_{1}=-1.3$ while bidding other parameters truthfully at the beginning. Assume also that agent $1$ always bids its state as if there is no noise ($w_{1}(t)\equiv 0$). Now agent 1's net expected utility is $0.631$. Therefore, agent $1$'s optimal strategy is not to bid its true state when it is allowed to bid its system parameters at the beginning.$\hfill\square$
\end{exmp}

The assumption that the ISO knows the system parameters $A$, $B$, $Q$ and $R$ of the agents can perhaps be justified since the ISO can learn these parameters by running a VCG scheme for the day-ahead market, a dynamic deterministic market, where agents are guaranteed to bid their true system parameters as shown in the previous section, and  system parameters remain unchanged when agents participate
in the real-time stochastic market. 

\subsection{Budget Balance and Individual Rationality in LQG systems}
We extend the notion of scaling and the associated SVCG mechanism to the stochastic dynamic systems as follows. Consider the payments
\begin{align*}
&p_{i}(s):=c\cdot \sum_{j\ne i}\sum_{t=s}^{T-1}\bigg[q_{j}\hat{x}_{j}^{2}(s,t)+r_{j}u_{j}^{(i)2}(s,t)\\&+2q_{j}\left(\sum_{\tau=0}^{s-1}\hat{x}_{j}(\tau,t)\right) \hat{x}_{j}(s,t)+2r_{j}\left(\sum_{\tau=0}^{s-1}u_{j}^{(i)(\tau,t)}\right)u_{j}^{(i)}(s,t)\bigg]\\&-\sum_{j\ne i}\sum_{t=s}^{T-1}\bigg[q_{j}\hat{x}_{j}^{2}(s,t)+r_{j}u_{j}^{*2}(s,t)\\&+2q_{j}\left(\sum_{\tau=0}^{s-1}\hat{x}_{j}(\tau,t)\right)\hat{x}_{j}(s,t)+2r_{j}\left(\sum_{\tau=0}^{s-1}u_{j}(\tau,t)\right)u_{j}^{*}(s,t)\bigg],
\end{align*}
\noindent where $u_{j}^{(i)}(s,t)$ is the optimal solution to:
\begin{align*}
&\max \ \sum_{j\ne i}\sum_{t=s}^{T-1}\bigg[q_{j}x_{j}^{2}(s,t)+u_{j}^{2}(s,t)\\&+2q_{j}\left(\sum_{\tau=0}^{s-1}x_{j}(\tau,t)\right)x_{j}(s,t)+2r_{j}\left(\sum_{\tau=0}^{s-1}u_{j}(\tau,t)\right)u_{j}(s,t)\bigg]
\end{align*}
\noindent subject to
\begin{equation*}
x_{j}(s,t)=a_{j}x_{j}(s,t-1)+b_{j}u_{j}(s,t-1),\text{ for }s<t\le T-1,
\end{equation*}
\begin{equation*}
\sum_{j\ne i}u_{j}(s,t)=0,\text{ for }s\le t\le T-1,
\end{equation*}
\begin{equation*}
x_{j}(s,s)=\hat{x}_{j}(s,s).
\end{equation*}
As in the static case, based on its prior knowledge of a suitable range for $c$, the ISO can choose a range of $c$, which does not depend on the agents' bids, to achieve BB and IR.

Truth-telling is a dominant strategy under the SVCG mechanism because it falls under the Groves mechanism. Under the dominant strategy equilibrium, every agent $i$ will bid its true state $x_{i}(s,s)$, i.e., $w_{i}(s-1)$.
\begin{theorem}
Let $U^{*}(t)$ be the optimal solution to the following problem:
\begin{equation*}
\max \ \mathbb{E}\sum_{t=0}^{T-1}[X^{T}(t)QX(t)+U^{T}(t)RU(t)]
\end{equation*}
\noindent subject to
\begin{equation*}
X(t+1)=AX(t)+BU(t)+W(t),
\end{equation*}
\begin{equation*}
1^{T}U(t)=0,\forall t,
\end{equation*}
\begin{equation*}
X(0)\sim\mathcal{N}(0,Z), W\sim\mathcal{N}(0,\Sigma).
\end{equation*}
Let $X^{(i)}(t):=[x_{1}(t),...,x_{i-1}(t),x_{i+1}(t),...x_{N}(t)]^{T}$, and similarly let $Q^{(i)}$, $R^{(i)}$, $A^{(i)}$, $B^{(i)}$, $Z^{(i)}$ and $\Sigma^{(i)}$ be the matrix with the $i$-th component removed. Let $U^{(i)}(t)$ be the optimal solution to the following problem:
\begin{equation*}
\max \ \mathbb{E}\sum_{t=0}^{T-1}[X^{(i)T}(t)Q^{(i)}X^{(i)}(t)+U^{T}(t)R^{(i)}U(t)]
\end{equation*}
\noindent subject to
\begin{equation*}
X^{(i)}(t+1)=A^{(i)}X^{(i)}(t)+B^{(i)}U(t)+W^{(i)}(t),
\end{equation*}
\begin{equation*}
1^{T}U(t)=0,\forall t,
\end{equation*}
\begin{equation*}
X^{(i)}(0)\sim\mathcal{N}(0,Z^{(i)}), W^{(i)}\sim\mathcal{N}(0,\Sigma^{(i)}).
\end{equation*}
Let $H_{i}:=\mathbb{E}\sum_{t=0}^{T-1}[X^{(i)T}(t)Q^{(i)}X^{(i)}(t)+U^{(i)T}(t)R^{(i)}U^{(i)}(t)]$ and let $H_{max}:=\max_{i}H_{i}$. Let $F_{total}=\mathbb{E}\sum_{t=0}^{T-1}[X^{T}(t)QX(t)+U^{T}(t)RU(t)]$. If $F_{total}>0$, $H_{i}>0$ for all $i$, and MPB \eqref{Market power balance} condition holds, there exists an $\underline{c}$ and $\bar{c}$, with $\underline{c}\le\bar{c}$ such that if the constant $c$ is chosen in the
range $[\underline{c},\ \bar{c}]$, then the SVCG mechanism for the deterministic
dynamic system satisfies IC, EF, BB and IR at the same time.
\end{theorem}
\begin{proof}
It is straightforward that under a dominant strategy,
\begin{equation*}
\mathbb{E}\bigg[\sum_{s=0}^{T-1}p_{i}(s)\bigg]=c\cdot H_{i}-\mathbb{E}\left[\sum_{j\ne i}\sum_{t=0}^{T-1}\left(q_{j}x_{j}^{2}(t)+r_{j}u_{j}^{*2}(t)\right)\right]
\end{equation*}
\noindent since $w_{i}$'s are i.i.d. and $u_{i}(t)$ is linear in $x_{i}(t)$. Hence, to achieve budget balance, we need,
\begin{equation*}
\mathbb{E}\bigg[\sum_{i}\sum_{s=0}^{T-1}p_{i}(s)\bigg]=c\cdot \sum_{i}H_{i}-(N-1)F_{total}\ge 0.
\end{equation*}

To achieve individual rationality for agent $i$, we need
\begin{align*}
&\mathbb{E}\left[\sum_{t=0}^{T-1}\left(q_{i}x_{i}^{2}(t)+r_{i}u_{i}^{*2}(t)\right)-\sum_{s=0}^{T-1}p_{i}(s)\right]\\
=&F_{total}-c\cdot H_{i}\ge 0.
\end{align*}

Combining both inequalities, we have
\begin{equation*}
\frac{(N-1)F_{total}}{\sum_{i}H_{i}}\le c \le \frac{F_{total}}{H_{max}}.
\end{equation*}
Let $\underline{c}=\frac{(N-1)F_{total}}{\sum_{i}H_{i}}$ and $\bar{c}=\frac{F_{total}}{H_{max}}$. To ensure $\underline{c}\le \bar{c}$, one sufficient condition is,
\begin{equation*}
(N-1)H_{max}\le \sum_{i}H_{i},\ F_{total}>0,\ H_{i}>0\text{ for all }i.\qedhere
\end{equation*}
\end{proof}

\subsection{Lagrange Optimality in LQG Systems}\label{cn}
In general, just as for a static problem, the SVCG mechanism is not Lagrange optimal. Within the feasible range $[\underline{c},\bar{c}]$, one can choose a $c$ that achieves near Lagrange optimality. This can be formulated as a MinMax problem:
\begin{equation*}
\min_{c}\max_{i}\ \left|\frac{d_{i}(c)}{\mathbb{E}\sum_{t=0}^{T-1}\left[\lambda^{*}(t)u_{i}^{*}(t)\right]}\right|,\text{ subject to \eqref{cbound}},
\end{equation*} 
where 
\begin{align*}
&d_{i}(c):=\mathbb{E}\sum_{t=0}^{T-1}\left[\lambda^{*}(t)u_{i}^{*}(t)-p_{i}(t)\right]\\
&=\mathbb{E}\sum_{t=0}^{T-1}\left[\lambda^{*}(t)u_{i}^{*}(t)\right]-c\cdot H_{i}+\mathbb{E}\left[\sum_{t=0}^{T-1}\left(q_{j}x_{j}^{2}(t)+r_{j}u_{j}^{*2}(t)\right)\right].
\end{align*} 
\begin{comment}
The MinMax problem can be transformed to a linear program:
\begin{equation*}
\min\ Z
\end{equation*}
\noindent subject to
\begin{equation*}
Z\ge \frac{d_{i}(c)}{\mathbb{E}\sum_{t=0}^{T-1}\left[\lambda^{*}(t)u_{i}^{*}(t)\right]}, \text{ for all }i,
\end{equation*}
\begin{equation*}
Z\ge -\frac{d_{i}(c)}{\mathbb{E}\sum_{t=0}^{T-1}\left[\lambda^{*}(t)u_{i}^{*}(t)\right]}, \text{ for all }i,
\end{equation*}
\begin{equation*}
\frac{(N-1)F_{total}}{\sum_{i}H_{i}}\le c\le\frac{F_{total}}{H_{max}}.
\end{equation*}
\end{comment}
\begin{exmp}
Consider the same system parameters as in Example \ref{bidabqr}. The optimal solution to the MinMax problem is $(c^{*},Z^{*})=(0.96,0.21)$. Thus, by choosing $c=0.96$, the SVCG mechanism satisfies IC, EF, BB and IR, and all agents expect to pay/receive within $21\%$ of their expected Lagrange optimal payments.
\end{exmp}

Just as for deterministic systems, as the number of agents increases, the scaled-VCG mechanism does achieve asymptotic Lagrange Optimality. 
\begin{theorem}
If $(a_{k}, b_{k}, p_{k}, q_{k}, \zeta_{k}, \sigma_{k})$ for $1 \leq k \leq N$ satisfy the following:
\begin{enumerate}
\item $\underline{a}\le |a_{i}|\le \bar{a}$, $\underline{b}\le |b_{i}|\le \bar{b}$, $\underline{q}\le q_{i}\le \bar{q}<0$ and $\underline{r}\le r_{i}\le \bar{r}<0$

\item $F_{total}>0$, $H_{i}>0$ for all $i$, and MPB condition holds,

\end{enumerate}
then the following holds:
\begin{enumerate}
\item There is a range within which $c^{N}$ can be chosen to achieve BB and IR, and $\lim_{N\to\infty}c^{N}=1$,

\item Asymptotic Lagrange Optimality: $\lim_{N\to\infty}\mathbb{E}\sum_{t=0}^{T-1}\left[\lambda^{N}(t)u_{i}^{N}(t)-p_{i}^{N}(t)\right]=0$, where the random variable $\lambda^{N}(t)$ is the Lagrange multiplier corresponding to the power balance constraint.
\end{enumerate}
\end{theorem}
\begin{proof}
At each layer, the ISO is solving a deterministic LQR problem, and from Theorem \ref{theo3} we have,
\begin{align*}
\lim_{N\to\infty}\frac{L_{s}^{*}}{H_{s,1}}=1, \text{ for }0\le s\le T-1,
\end{align*}
where $L_{s}^{*}$ is the maximum value of $L_{s}$, and $H_{s,1}$ is the maximum social welfare when agent $1$ is excluded. Moreover, as we have shown in Theorem \ref{stochasticvcg}, the sum of $U^{*}(s,t)$ calculated at each layer is indeed the optimal solution $U^{*}(t)=\sum_{s=0}^{t}U^{*}(s,t)$.
Consequently,
\begin{align*}
\lim_{N\to\infty}\frac{F_{total}^{N}}{H_{1}^{N}}=\lim_{N\to\infty}\frac{\mathbb{E}\sum_{s=0}^{T-1}L_{s}^{*}}{\mathbb{E}\sum_{s=0}^{T-1}H_{s,1}}=1.
\end{align*}
Similarly we can show that $\lim_{N\to\infty}\frac{F_{total}^{N}}{H_{k}^{N}}=1$ for $k\ne 1$. Therefore, $\lim_{N\to\infty}\bar{c}^{N}=1$. Let $H_{min}=\min H_{i}$. Since $\frac{(N-1)F_{total}^{N}}{NH_{max}}\le \underline{c}^{N}\le \frac{(N-1)F_{total}^{N}}{NH_{min}}$, $\lim_{N\to\infty}\underline{c}^{N}=1$. Hence, $\displaystyle\lim_{N\to\infty}c^{N}=1.$
We next show that the total expected VCG payment converges to the total expected Lagrange payment when $N$ goes to infinity. To calculate $\lambda^{N}(t)$, we solve the following one-step problem to determine its Lagrange multiplier:  
\begin{equation*}
\max_{U(t)}\ X^{T}(t)QX(t)+U^{T}(t)RU(t)+\mathbb{E}\left[X^{T}(t+1)P_{t+1}X(t+1)\right]
\end{equation*}
\noindent subject to
\begin{equation}\label{B}
1^{T}U(t)=0.
\end{equation}
\noindent where $P_{t}$ is the Ricatti matrix of the unconstrained problem where balance constraint $1^{T}U(t)=0$, or $u_{1}(t)=-\sum_{i=2}^{N}u_{i}(t)$ is substituted in both the objective and the state equation. 

The Lagrangian is,
\begin{align*}
\mathcal{L}=&X^{T}(t)QX(t)+U^{T}(t)RU(t)\\
&+\mathbb{E}\left[X^{T}(t+1)P_{t+1}X(t+1)\right]-\lambda^{N}(t)1^{T}U(t)
\end{align*}
Take partial derivative with respect to $U(t)$ and $\lambda(t)$, we have
\begin{align*}
&\frac{\partial\mathcal{L}}{\partial U(t)}\\
=&2RU(t)+2B^{T}P_{t+1}BU(t)+2B^{T}P_{t+1}AX(t)-\lambda^{N}(t)1=0.
\end{align*}
The Lagrange multiplier $\lambda^{N}(t)$ calculated from \eqref{B} is:
\begin{multline}\label{stochasticlambda}
\lambda^{N}(t)=2\left[1^{T}(R+B^{T}P_{t+1}B)^{-1}1 \right]^{-1}\cdot\\
1^{T}(R+B^{T}P_{t+1}B)^{-1}B^{T}P_{t+1}AX(t):=\Phi_{t}X(t).
\end{multline}
At time $s$, we denote $\lambda^{N}(s,t)$ as the Lagrange multipliers associated with the balance constraint $1^{T}U(s,t)=0$. From Theorem \ref{theo3}, we have
\begin{align*}
\lim_{N\to\infty}\left[\left(\sum_{t=s}^{T-1}\lambda^{N}(s,t)u_{i}^{*N}(s,t)\right)-p_{i}^{N}(s)\right]=0.
\end{align*}
Summing over $s$, we have
\begin{align*}
&\lim_{N\to\infty}\sum_{s=0}^{T-1}\left[\left(\sum_{t=s}^{T-1}\lambda^{N}(s,t)u_{i}^{*N}(s,t)\right)-p_{i}^{N}(s)\right]\\
=&\lim_{N\to\infty}\sum_{t=0}^{T-1}\left[\left(\sum_{s=0}^{t}\lambda^{N}(s,t)u_{i}^{*N}(s,t)\right)-p_{i}^{N}(t)\right]=0.
\end{align*}
From \eqref{completionofsquare}, we have, at time $s$, $\displaystyle\lambda^{N}(s,t)=\Phi_{t}\sum_{\tau=0}^{s}X(\tau,t),$ and at time $s-1$, $\displaystyle\lambda^{N}(s-1,t)=\Phi_{t}\sum_{\tau=0}^{s-1}X(\tau,t).$
Therefore, $\lambda^{N}(s,t)=\lambda^{N}(s-1,t)+\Phi_{t}X(s,t).$

The Lagrange multiplier $\lambda^{N}(t)$ associated with the balance constraint $1^{T}U(t)=0$ can be calculated as:
\begin{align*}
\lambda^{N}(t)=\Phi_{t}X(t)=\Phi_{t}\sum_{s=0}^{t}X(s,t)=\lambda^{N}(t,t).
\end{align*}
As a result,
\begin{align*}
&\lambda^{N}(t)u_{i}^{*N}(t)=\lambda^{N}(t)\sum_{s=0}^{t}u_{i}^{*N}(s,t)\\
&=\sum_{s=0}^{t}\left[\left(\lambda^{N}(s,t)+\Phi_{t}\sum_{\tau=s+1}^{t}X(\tau,t)\right)u_{i}^{*N}(s,t)\right]
\end{align*}
Because $X(0)$ is independent of $W(t)$ and $W(t)$ are i.i.d., 
\begin{align*}
\mathbb{E}[X(\tau,t)u_{i}^{*N}(s,t)]=0, \text{ for }\tau\ge s+1.
\end{align*}
Therefore,
\begin{align*}
&\lim_{N\to\infty}\mathbb{E}\sum_{t=0}^{T-1}\left[\lambda^{N}(t)u_{i}^{*N}(t)-p^{N}_{i}(t)\right]\\
&=\lim_{N\to\infty}\mathbb{E}\sum_{t=0}^{T-1}\left[\left(\sum_{s=0}^{t}\lambda^{N}(s,t)u_{i}^{*N}(s,t)\right)-p_{i}^{N}(t)\right]=0.
\end{align*}
\end{proof}
\subsection{Numerical Example}
In this section, we provide a numerical example of the proposed Scaled-VCG scheme. For power generators, $u_{i}(t)$ is the power generated at time $t$. The state equation is:
\begin{equation*}
x_{i}(t+1)=x_{i}(t)+u_{i}(t)+w_{i}(t),
\end{equation*}
where $w_{i}(t)$ is the noise. The utility function is defined as the negative cost: 
\begin{equation*}
F_{i}(u_{i}(t))=r_{i}u_{i}^{2}(t)+v_{i}u_{i}(t),
\end{equation*}
where $r_{i}<0$ and $v_{i}>0$ are constants.
For power consumers, we adopt the virtual battery model \cite{7525121} to represent aggregate load, and let $x(t)$ and $u(t)$ denote the state of charge (SoC) and power consumption at time $t$, respectively. The state equation is: 
\begin{equation*}
x_{i}(t+1)=a_{i}x_{i}(t)+b_{i}u_{i}(t)+\eta_{i}h_{i}(t)+w_{i}(t),
\end{equation*}
where $a_{i}$, $b_{i}$ and $\eta_{i}$ are constants, $h_{i}(t)$ represents ambient temperature, and $w_{i}(t)$ is the noise. The utility function is:
\begin{equation*}
F_{i}(x_{i}(t),u_{i}(t))=q_{i}(x_{i}(t)-\bar{x}_{i})^2,
\end{equation*} 
where $q_{i}<0$ is a constant number, and $\bar{x}_{i}$ is the desired SoC.

Let $T=4$. The parameters are drawn from the following uniform distributions: $r_{i}\sim\mathcal{U}(-0.02,-0.01)$, $v_{i}\sim\mathcal{U}(40,50)$, $a_{i}\sim\mathcal{U}(0.99,1.01)$, $b_{i}\sim\mathcal{U}(0.98,1.02)$, $\eta_{i}\sim\mathcal{U}(3.5,4)$, $q_{i}\sim\mathcal{U}(-1.01,-0.99)$. $w_{i}(t)$ are drawn from Gaussian distribution $\mathcal{N}(0,0.2)$. For generators, $x_{i}(0)=0$; for consumers, $x_{i}(t)\sim\mathcal{N}(75,1)$ and $h_{i}=[75,80,90,75]^{T}, \forall i$. We first investigate the issue of incentive compatibility in a 4-agent system with agents 1 and 2 as generators, and agents 3 and 4 as loads. We assume that agents 2, 3, and 4 always bid the true states while agent 1 may adopt different strategies when bidding its state. We compare the expected net utility for agent 1 when the following 3 strategies are adopted:
\begin{enumerate}
\item Always bids the state $\hat{x}_{i}(s,s)=w_{i}(s-1)-0.1$.
\item Always bids the true state $\hat{x}_{i}(s,s)=w_{i}(s-1)$.
\item Always bids the state $\hat{x}_{i}(s,s)=w_{i}(s-1)+0.1$.
\end{enumerate}
The results are summarized in Table \ref{coms}. It can be seen that truth-telling of state results in a higher expected net utility compared to other non-truth-telling strategies. 
\begin{table}
\normalsize
\centering
\begin{tabular}{|c|c|c|c|}
\hline
Strategy & 1 & 2 & 3 \\
\hline
Expected net utility & 25.4 & 27.8 & 25.2\\
\hline
\end{tabular}
\caption{Comparison of expected net utility when adopting different strategies}
\label{coms}
\end{table}

We next investigate the asymptotic performance of the SVCG mechanism. For each $N$, $c^{N}$ is chosen as the optimal solution to the MinMax problem as shown in Section \ref{cn}. Denote $d^{N}=\mathbb{E}\sum_{t=0}^{T-1}\left[\lambda^{N}(t)u_{i}^{N}(t)-p_{i}^{N}(t)\right]$. Results are shown in Fig. \ref{numcn}. It is seen that as $N$ increases, $c^{N}$ and $d^{N}$ converge to $1$ and $0$ respectively.
\begin{figure}[t]
	\centering
	\captionsetup{justification=centering}	
	\includegraphics[width=\linewidth]{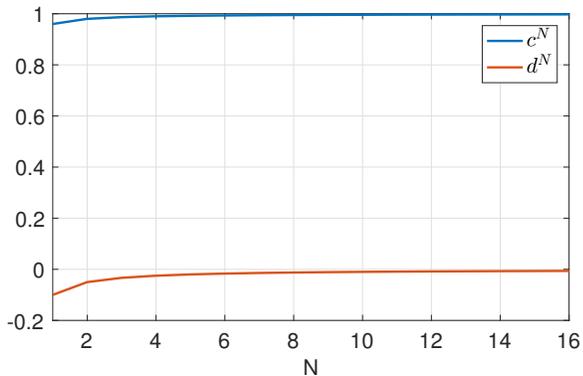}
	\caption{Convergence of $c^{N}$ and $d^{N}$}
	\label{numcn}
\end{figure}
\subsection{LQG systems with partially observed states}
The above results can be extended in a straightforward manner to the partially observed case where a linear transformation of the state is observed under additive white Gaussian noise. This can be done by considering the hyperstate, the conditional distribution of the state, which evolves as a linear system driven by the innovation process. In this case, the state noise corresponds to the innovation process, which is effectively what the agents are being asked to bid.
\section{Linear Non-Gaussian Agents with Quadratic Costs}\label{NG}
Consider now the case where the agents are linear and still have quadratic costs, but their state noises, while independent, are not Gaussian.
\begin{defn}
We say that a mechanism achieves ``linear efficiency" if it attains the best cost that can be obtained through linear feedback.
\end{defn}

The mechanism for the LQG case also achieves linear efficiency when the noises are zero mean and white, but non-Gaussian. This result rests on the fact that the mechanism achieves intertemporal decoupling.
\begin{theorem}
The SVCG mechanism can ensure incentive compatibility and linear efficiency.
\end{theorem}
\begin{proof}
The proof follows the the proof of Theorem \ref{stochasticvcg} for LQG since the crucial idea there is to obtain uncorrelatedness as in Lemma \ref{decoupling}, which follows from the linear control law.
\end{proof}
In general, however, nonlinear strategies may achieve even lower cost, but the SVCG mechanism can only guarantee the optimal cost in the class of linear feedback laws.
\section{Concluding Remarks}\label{CRCDC}

It remains an open problem if it is possible to construct a mechanism that ensures the dominance of dynamic truth-telling for agents comprised of general stochastic dynamic systems. We conjecture that it is not feasible in general. A careful construction of a sequence of layered VCG payments over time shows that for the special case of rational LQG agents with known
system parameters, the intertemporal effect of current bids on future payoffs can be decoupled, and truth-telling of dynamic states is guaranteed. It achieves subgame perfect dominance of truth telling strategies and social welfare optimality. A modification of the VCG payments, called Scaled-VCG, achieves Budget Balance and Individual Rationality for a range of scaling factors, under a certain Market Power Balance condition. This condition provides economic justification for Load Serving entities or Load Aggregators that group small consumers as a means for achieving social welfare optimality. If the noises are not Gaussian, then the mechanism achieves optimal social welfare in the class of linear strategies. In the asymptotic regime of increasing population of agents, the Scaled-VCG payments converge to the Lagrange payments, the payments that the agents would make in the absence of strategic considerations. It is of interest to determine the viability of the ISO identifying a range of the scaling constant that assures budget balance and individual rationality. It is also of interest to design incentive mechanisms or identify conditions under which there is no strategic play on the scaling
constants. and to extend the current layered mechanism to the case where the ISO does
not know the system parameters.

\section*{ACKNOWLEDGMENT}
The authors would like to thank Le Xie and Vijay Subramanian for their valuable comments and suggestions.
\begin{comment}
However, there is also a fatal downside for the VCG mechanism: in general, the sum of
total payments collected by the ISO may be negative. In fact, when agents have quadratic utility
functions, the total payments collected from consumers is
indeed not enough to cover the total payments to the suppliers.
In effect, in order to force agents to reveal their true utility
functions, the ISO needs to subsidize the market. It remains to be seen that how the classic VCG mechanism could be modified to ensure budget balance for the ISO. It also remains unclear whether individual agents will be actively participating in the market since payments charged by the ISO may result in a negative net utility for agents.
\end{comment}

\bibliographystyle{IEEEtran}
\bibliography{myReference}

\begin{IEEEbiography}[{\includegraphics[width=1in,height=1.25in,clip,keepaspectratio]{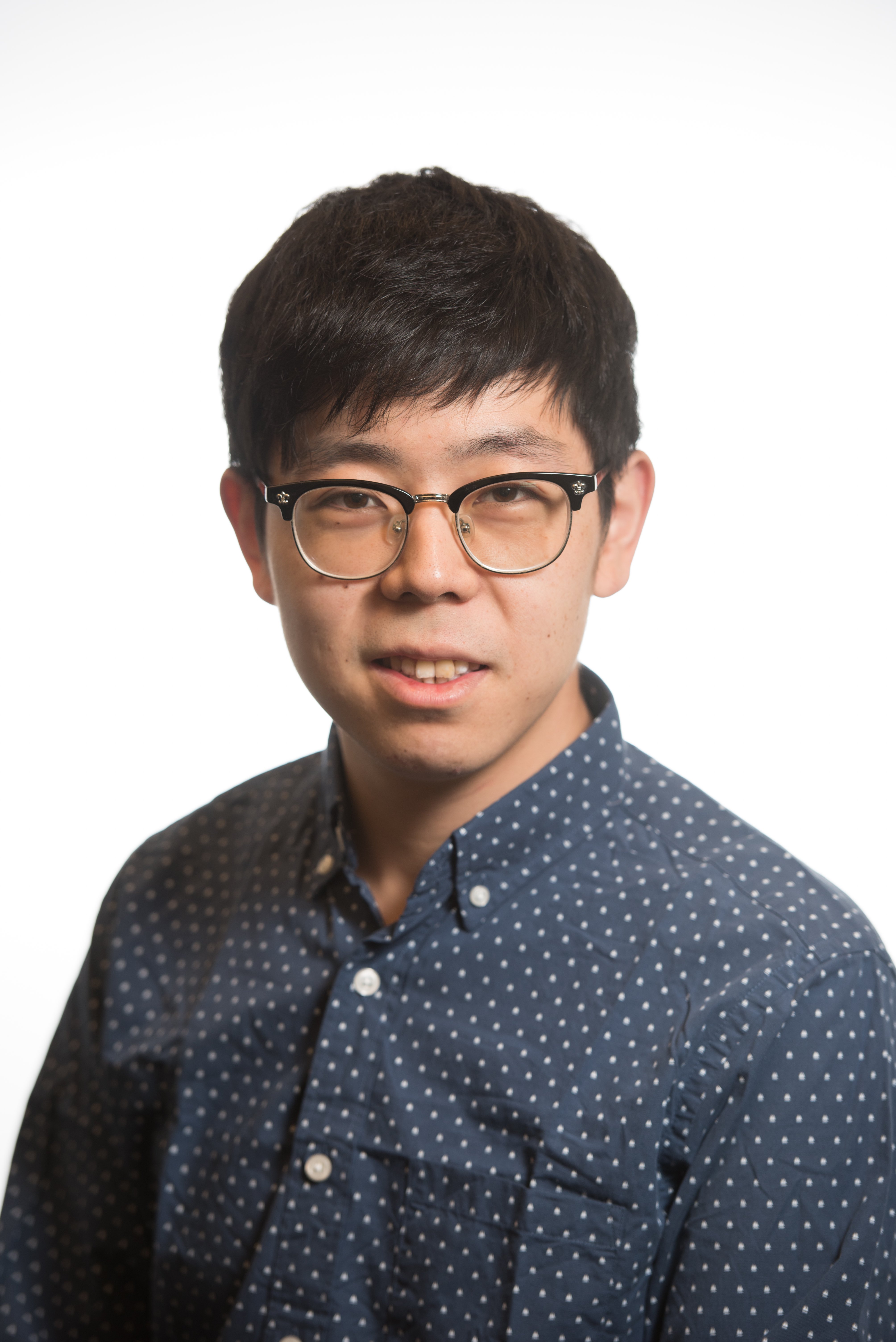}}]{Ke Ma}
received the B.E. degree in automation from Tsinghua University, Beijing, China, in 2012, and the Ph.D. degree in electrical and computer engineering from the Department of Electrical and Computer Engineering, Texas A\&M University, College Station, TX, USA in 2018. 

He is currently an electrical engineer at the Optimization and Control Group, Pacific Northwest National Laboratory (PNNL), Richland, WA, USA. His research interests include dynamic mechanism design and its application in electricity market, and market-based (transactive) coordination and control of distributed energy resources (DERs).
\end{IEEEbiography}
\begin{IEEEbiography}[{\includegraphics[width=1in,height=1.25in,clip,keepaspectratio]{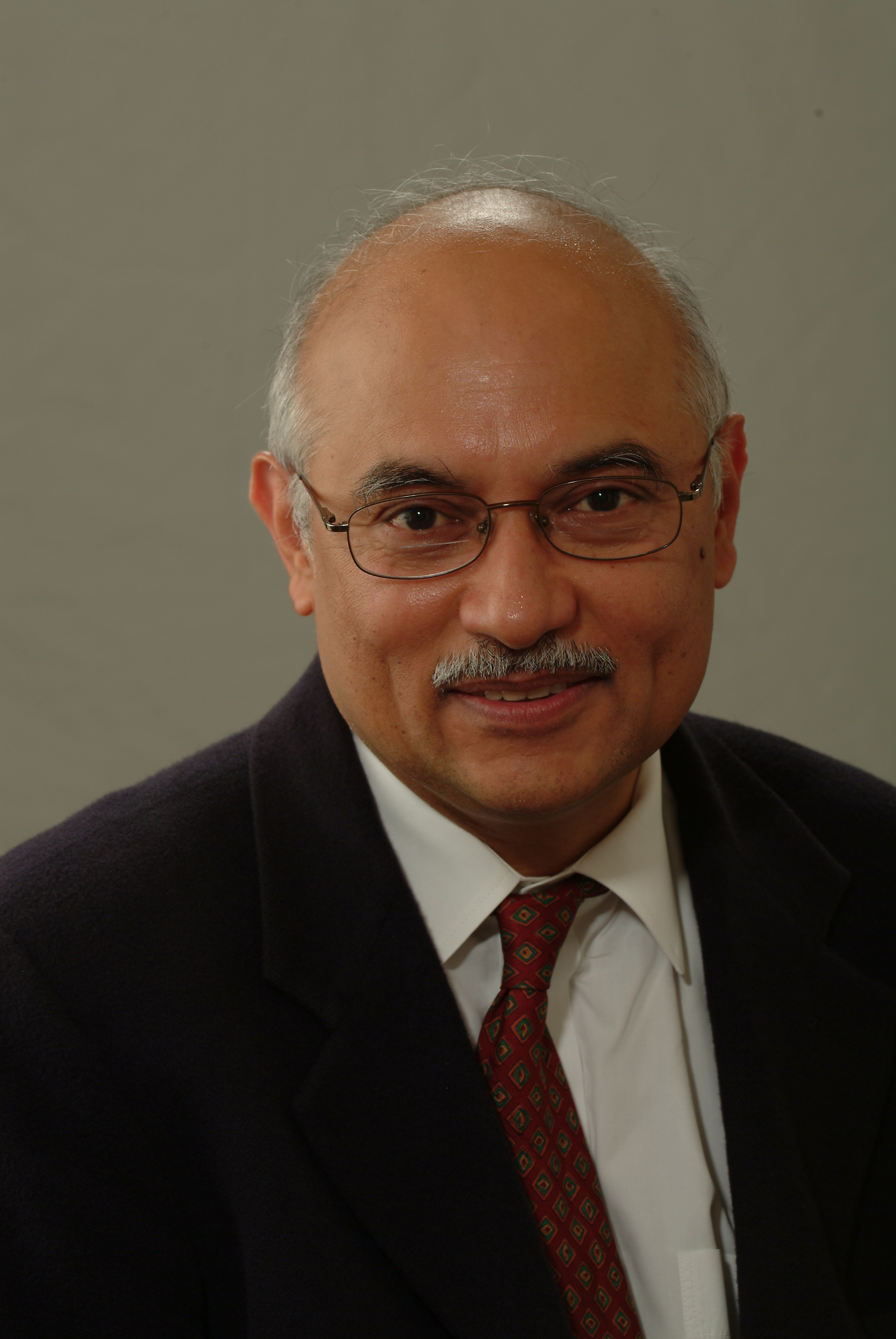}}]{P.~R.~Kumar} B. Tech. (IIT Madras, `73), 
D.Sc. (Washington University, St.~Louis, `77), was a faculty member at UMBC (1977-84) and Univ.~of Illinois, Urbana-Champaign (1985-2011). He is currently at Texas A\&M University. His current research is focused on cyberphysical systems, cybersecurity, privacy, wireless networks, renewable energy, power system, smart grid, autonomous vehicles, and unmanned air vehicle systems. 

He is a member of the US National Academy of Engineering, The World Academy of Sciences, and the Indian National Academy of Engineering. He was awarded a Doctor Honoris Causa by ETH, Zurich. He has received the IEEE Field Award for Control Systems, the Donald~P.~Eckman Award of the AACC,  Fred~W.~Ellersick Prize of the IEEE Communications Society, the Outstanding Contribution Award of ACM SIGMOBILE, the Infocom Achievement Award, and the SIGMOBILE Test-of-Time Paper Award. He is a Fellow of IEEE and ACM Fellow. He was Leader of the Guest Chair Professor Group on Wireless Communication and Networking at Tsinghua University, is a D. J. Gandhi Distinguished Visiting Professor at IIT Bombay, and an Honorary Professor at IIT Hyderabad. He was awarded the Distinguished Alumnus Award from IIT Madras, the Alumni Achievement Award from Washington Univ., and the Daniel Drucker Eminent Faculty Award from the College of Engineering at the Univ.~of Illinois.
\end{IEEEbiography}
\end{document}